%% file: main.tex
\documentclass[runningheads]{llncs}
\usepackage[utf8]{inputenc}
\usepackage{graphicx}
\usepackage[draft]{todonotes}
\usepackage{booktabs}
\usepackage{array}
\usepackage{amsmath}

\usepackage{amsthm}
\usepackage{cite}
\usepackage{xspace}
\usepackage{wrapfig}
\usepackage{subcaption}
\usepackage[hyphens]{url}
\usepackage{hyperref}
\usepackage{cleveref}
\usepackage{makecell}

\hypersetup{breaklinks=true}
\graphicspath{{figures/}}

\newcolumntype{L}{l<{\hspace{1cm}}}
\newcolumntype{T}{c<{\hspace{1cm}}}
\newcolumntype{C}{c<{\hspace{.8cm}}}

\begin{document}
    \title{Fraud and Data Availability Proofs: Maximising Light Client Security and Scaling Blockchains with Dishonest Majorities}
    \titlerunning{Fraud and Data Availability Proofs}

    \author{
        Mustafa Al-Bassam\inst{1}\and
        Alberto Sonnino\inst{1}\and
        Vitalik Buterin\inst{2}
    }
    \institute{
        University College London\\
        \email{\{m.albassam,a.sonnino\}@cs.ucl.ac.uk}\and
        Ethereum Research\\
        \email{vitalik@ethereum.org}
    }

    \maketitle

    \input{macros}
    \input{sections/abstract}
    \input{sections/introduction}
    \input{sections/background}
    \input{sections/model}
    \input{sections/fraud-proofs}
    \input{sections/availability}
    \input{sections/implementation}
    \input{sections/discussion}
    \input{sections/related-work}
    \input{sections/conclusion}
    \input{sections/acknowledgements}

    \bibliographystyle{splncs04}
    \bibliography{references.bib}

    \input{sections/appendix}
\end{document}

%% file: macros.tex
% Comment macros
%\newcommand\alberto[1]{\todo[color=yellow]{\textbf{Alberto:} #1}}
%\newcommand\mustafa[1]{\todo[color=green]{\textbf{Mustafa:} #1}}
%\newcommand\vitalik[1]{\todo[color=brown]{\textbf{Vitalik:} #1}}
\newcommand\albertoi[1]{\todo[color=yellow,inline]{\textbf{Alberto:} #1}}
\newcommand\mustafai[1]{\todo[color=green,inline]{\textbf{Mustafa:} #1}}
\newcommand\vitaliki[1]{\todo[color=brown,inline]{\textbf{Vitalik:} #1}}

\newcommand\alberto[1]{}
\newcommand\mustafa[1]{}
\newcommand\vitalik[1]{}

% Abbreviations
\newcommand{\cf}{cf.\@\xspace}
\newcommand{\vs}{vs.\@\xspace}
\newcommand{\etc}{etc.\@\xspace}
\newcommand{\ala}{ala\@\xspace}
\newcommand{\wrt}{w.r.t.\@\xspace}
\newcommand{\etal}{\textit{et al.}\@\xspace}
\newcommand{\eg}{\textit{e.g.,}\@\xspace}
\newcommand{\ie}{\textit{i.e.,}\@\xspace}
\def\first{({\it i})\xspace}
\def\second{({\it ii})\xspace}
\def\third{({\it iii})\xspace}
\def\fourth{({\it iv})\xspace}
\def\fifth{({\it v})\xspace}
\def\sixth{({\it vi})\xspace}

%% file: sections/abstract.tex
\begin{abstract}
    Light clients, also known as Simple Payment Verification (SPV) clients, are nodes which only download a small portion of the data in a blockchain, and use indirect means to verify that a given chain is valid. Typically, instead of validating block data, they assume that the chain favoured by the blockchain's consensus algorithm only contains valid blocks, and that the majority of block producers are honest. By allowing such clients to receive fraud proofs generated by fully validating nodes that show that a block violates the protocol rules, and combining this with probabilistic sampling techniques to verify that all of the data in a block actually is available to be downloaded, we can eliminate the honest-majority assumption for block validity, and instead make much weaker assumptions about a minimum number of honest nodes that rebroadcast data. Fraud and data availability proofs are key to enabling on-chain scaling of blockchains (\eg via sharding or bigger blocks) while maintaining a strong assurance that on-chain data is available and valid. We present, implement, and evaluate a novel fraud and data availability proof system.
\end{abstract}

%% file: sections/introduction.tex
\section{Introduction and Motivation}\label{sec:introduction}

As cryptocurrencies and smart contract platforms have gained wider adoption, the scalability limitations of existing blockchains have been observed in practice. Popular services have stopped accepting Bitcoin \cite{nakamoto2008} payments due to transactions fees rising as high as \$20 \cite{orland2017, karlo2018}, and Ethereum's \cite{buterin2013} popular CryptoKitties smart contract caused the pending transactions backlog to increase six-fold \cite{wong2017}. Users pay higher fees as they compete to get their transactions included on the blockchain, due to on-chain space being limited, \eg by Bitcoin's block size limit \cite{antonopoulos2014} or Ethereum's block gas limit \cite{wood2018}.

While increasing on-chain capacity limits would yield higher transaction throughput, there are concerns that this would decrease decentralisation and security, because it would increase the resources required to fully download and validate the blockchain, and thus fewer users would be able to afford to run full nodes that independently validate the blockchain, requiring users to instead run light clients that assume that the chain favoured by the blockchain's consensus algorithm abides by the protocol rules for transaction validity \cite{marshall2017}. Light clients operate well under normal circumstances, but have weaker assurances when the majority of the consensus (\eg miners or block producers) is dishonest; for example, whereas a dishonest majority in the Bitcoin or Ethereum network can at present only censor, reverse or reorder transactions, if all clients are using light nodes, a majority of the consensus would be able to collude together to generate blocks that contain contain invalid transactions that, for example, create money out of thin air, and light nodes would not be able to detect this. On the other hand, full nodes would reject those invalid blocks immediately.

As a result, various scalability efforts have focused on off-chain scaling techniques such as payment channels \cite{poon2016}, where participants sign transactions off-blockchain, and settle the final balance on-chain. Payment channels have also been generalised to state channels \cite{miller2017}. However, as opening and settling channels involves on-chain transactions, on-chain scaling is still necessary for widespread adoption of payment and state channels.\footnote{Suppose a setting where all users used channels and channels only needed to be opened once and maintained with on-chain transactions once per year per used. To support a userbase equal in size to Facebook's ($\approx2.2$ billion \cite{nieva2018}), one would need 2.2 billion transactions per year, or $\approx70$ transactions per second, significantly higher than supported by the Bitcoin or Ethereum blockchains \cite{croman2016, young2018}. This does not take into account usages that require ``going on-chain'' more frequently, users requiring multiple channels, or the possibility of attacks on channels requiring more transactions to process.}

In this paper, we decrease the on-chain capacity \vs security trade-off by making it possible for light clients to receive and verify fraud proofs of invalid blocks from full nodes, so that they too can reject them, assuming that there is at least one honest full node willing to generate fraud proofs to be propagated within a maximum network delay. We also design a data availability proof system, a necessary complement to fraud proofs, so that light clients have assurance that the block data required for full nodes to generate fraud proofs from is available, given that there is a minimum number of honest light clients to reconstruct missing data from blocks. We implement and evaluate the security and efficiency of our overall design.

Our work also plays a key role in efforts to scale blockchains with sharding \cite{buterin2018, al-bassam2018, kokoris-kogias2018}, as in a sharded system no single node in the network is expected to download and validate the state of all shards, and thus fraud proofs are necessary to detect invalid blocks from malicious shards.

%% file: sections/background.tex
\section{Background}\label{sec:background}

\subsection{Blockchain Models}

Briefly, the data structure of a blockchain consists of (literally) a chain of blocks. Each block contains two components: a header and a list of transactions. In addition to other metadata, the header stores at minimum the hash of the previous block (thus enabling the chain property), and the root of the Merkle tree that consists of all transactions in the block.

Blockchain networks have a consensus algorithm \cite{bano2017} to determine which chain should be favoured in the event of a fork, \eg if proof-of-work \cite{nakamoto2008} is used, then the chain with the most accumulated work is favoured. They also have a set of transaction validity rules that dictate which transactions are valid, and thus blocks that contain invalid transactions will never be favoured by the consensus algorthim and should in fact always be rejected.

Full nodes are nodes which download block headers as well as the list of transactions, verifying that the transactions are valid according to some transaction validity rules. Light clients only download block headers, and assume that the list of transactions are valid according to the transaction validity rules. Light clients verify blocks against the consensus rules, but not the transaction validity rules, and thus assume that the consensus is honest in that they only included valid transactions. Light clients can receive Merkle proofs from full nodes that a specific transaction or state object is included in a block header.

There are two major types of blockchain transaction models: Unspent Transaction Output (UTXO)-based, and account-based. Transactions in UTXO-based blockchains (\eg Bitcoin) contain references to previous transactions whose coins they wish to `spend'. As a single transaction may send coins to multiple addresses, a transaction has many `outputs', and thus new transactions contain references to these specific outputs. Each output can only be spent once.

On the other hand, account-based blockchains (\eg Ethereum), are somewhat simpler to work with (though sometimes more complex to apply parallelisation techniques to), as each transaction simply specifies a balance transfer from one address to another, without reference to previous transactions. In Ethereum, the block header also contains a root to a Merkle tree containing the state, which is the `current' information that is required to verify the next block; in Ethereum this consists of the balance, code and permanent storage of all of the accounts and contracts in the system.

\subsection{Merkle Trees and Sparse Merkle Trees}

A Merkle tree \cite{merkle1988} is a binary tree where every non-leaf node is labelled with the cryptographic hash of the concatenation of its children nodes. The root of a Merkle tree is thus a commitment to all of the items in its leaf nodes. This allows for Merkle proofs, which given some Merkle root, are proofs that a leaf is a part of the tree committed to by the root. A Merkle proof for some leaf consists of all of the ancestor and ancestor's sibling intermediate nodes for that leaf, up to the root of the tree, thus forming a sub-tree whose Merkle root can be recomputed to verify that the Merkle proof is valid. The size and verification time of a Merkle proof for a tree with $n$ leaves is $O(\log(n))$, as it is a tree.

A sparse Merkle tree \cite{laurie2012, rasmus2016} is a Merkle tree with $n$ leaves where $n$ is extremely large (\eg $n = 2^{256}$), but where almost all of the nodes have the same default value (\eg $0$). If $k$ nodes are non-zero, then at each intermediate level of the tree there will be a maximum of $k$ non-zero values, and all other values will be the same default value for that level: $0$ at the bottom level, $L_1 = H(0, 0)$ at the first intermediate level, $L_2 = H(L_1, L_1)$ at the second intermediate level, and so on. Hence, despite the exponentially large number of nodes in the tree, the root of the tree can be calculated in $O(k \times \log(n))$ time. A sparse Merkle tree allows for commitments to key-value maps, where values can be updated, inserted or deleted trivially in $O(\log(n))$ time. Merkle proofs of specific key-values entries are of size $\log(n)$ if constructed naively but can be compressed to size $\log(k)$ as intermediate nodes whose sibling have the default value do not need to explicitly be shown.

Systems such as Ripple and Ethereum at present use Patricia trees instead of sparse Merkle trees \cite{wood2018, schwartz2013}; we use sparse Merkle trees in this paper because of their greater simplicity.

\subsection{Erasure Codes and Reed-Solomon Codes}\label{sec:erasure-codes}

Erasure codes are error-correcting codes \cite{elias1954, peterson1972} working under the assumption of bit erasures rather than bit errors; in particular, the users knows which bits have to be reconstructed. Error-correcting codes transform a message of length $k$ into a longer message of length $n > k$ such that the original message can be recovered from a subset of the $n$ symbols.

Reed-Solomon (RS) codes \cite{wicker1994} have various applications and are among the most studied error-correcting codes. A Reed-Solomon code encodes data by treating a length-$k$ message as a list of elements $x_0, x_1, ..., x_{k-1}$ in some finite field (prime fields and binary fields are most frequently used), interpolating the polynomial $P(x)$ where $P(i) = x_i$ for all $0 \le i < k$, and then extending the list with $x_k, x_{k+1}, ..., x_{n-1}$ where $x_i = P(i)$. The polynomial $P$ can be recovered from any $k$ symbols from this longer list using techniques such as Lagrange interpolation, or more optimized and advanced techniques involving tools such as Fast Fourier transforms, and knowing $P$ one can then recover the original message. Reed-Solomon codes can detect and correct any combination of up to $\frac{n-k}{2}$ errors, or combinations of errors and erasures. RS codes have been generalised to multidimensional codes \cite{shea2003, dudavcek2016} in various ways \cite{shen1998, wu1992, saints1995}. In a $d$-dimensional code, the message is encoded into a square or cube or hybercube of size $k \times k \times ... \times k$, and a multidimensional polynomial $P(x_1, x_2, ..., x_d)$ is interpolated where $P(i_1, i_2, ..., i_n) = x_{i_1, i_2 ..., i_n}$, and this polynomial is extended to a larger $n \times n \times ... \times n$ square or cube or hypercube.

%% file: sections/model.tex
\section{Assumptions and Threat Model}\label{sec:threat-model}

We present the network and threat model under which our fraud proofs (\Cref{sec:fraud-proofs}) and data availability proofs (\Cref{sec:availability}) apply.

\subsection{Preliminaries}

We present some primitives that we use in the rest of the paper.

\begin{itemize}
    \item $\mathsf{hash}(x)$ is a cryptographically secure hash function that returns the digest of $x$ (\eg SHA-256).
    \item $\mathsf{root}(L)$ returns the Merkle root for a list of items $L$.
    \item $\{e \rightarrow r\}$  denotes a Merkle proof that an element $e$ is a member of the Merkle tree committed by root $r$.
    \item $\mathsf{VerifyMerkleProof}(e, \{e \rightarrow r\}, r, n, i)$ returns $\mathsf{true}$ if the Merkle proof is valid, otherwise $\mathsf{false}$, where $n$ additionally denotes the total number of elements in the underlying tree and $i$ is the index of $e$ in the tree. This verifies that $e$ is at index $i$, as well as its membership.
    \item $\{k, v \rightarrow r\}$  denotes a Merkle proof that a key-value pair $k, v$ is a member of the Sparse Merkle tree committed by root $r$.
    %\item $\mathsf{VerifySparseMerkleProof}(k, v, \{k, v \rightarrow r\}, r)$ returns $\mathsf{true}$ if the Sparse Merkle proof is valid, otherwise $\mathsf{false}$.
\end{itemize}

\subsection{Blockchain Model}\label{sec:blockchain-model}

We assume a generalised blockchain architecture, where the blockchain consists of a hash-based chain of block headers $H = (h_0, h_1, ..., h_n)$. Each block header $h_i$ contains a Merkle root $\mathsf{txRoot}_i$ of a list of transactions $T_i$, such that $\mathsf{root}(T_i) = \mathsf{txRoot}_i$. Given a node that downloads the list of transactions $N_i$ from the network, a block header $h_i$ is considered to be valid if \first $\mathsf{root}(N_i) = r_i$ and \second given some validity function
\begin{equation*}
    \mathsf{valid}(T, S) \in \{\mathsf{true}, \mathsf{false}\}
\end{equation*}
where $T$ is a list of transactions and $S$ is the state of the blockchain, then $\mathsf{valid}(T_i, S_{i-1})$ must return $\mathsf{true}$, where $S_i$ is the state of the blockchain after applying all of the transactions in $T_i$. We assume that $\mathsf{valid}(T, S)$ takes $O(n)$ time to execute, where $n$ is the number of transactions in $T$.

In terms of transactions, we assume that given a list of transactions $T_i = (t_i^0, t_i^1, ..., t_i^n)$, where $t_i^j$ denotes a transaction $j$ at block $i$, there exists a state transition function $\mathsf{transition}$ that returns the post-state $S'$ of executing a transaction on a particular pre-state $S$, or an error if the transition is illegal:
\begin{equation*}
    \mathsf{transition}(S, t) \in \{S', \mathsf{err}\}
\end{equation*}
\begin{equation*}
    \mathsf{transition}(\mathsf{err}, t) = \mathsf{err}
\end{equation*}

Thus given the intermediate post-states after applying every transaction one at a time, $I_i^j = \mathsf{transition}(I_i^{j-1}, t_i^j)$, and the base case $I_i^{-1} = S_{i-1}$, then $S_i = I_i^n$. Hence, $I_i^j$ denotes the intermediate state of the blockchain at block $i$ after applying transactions $t_i^0, t_i^1, ..., t_i^j$.

Therefore, $\mathsf{valid}(T_i, S_{i-1}) = \mathsf{true}$ if and only if $I_i^n \neq \mathsf{err}$.

In \Cref{sec:state-roots}, we explain how both a UTXO-based (\eg Bitcoin) and an account-based (\eg Ethereum) blockchain can be represented by this model.

\subsubsection{Aim.}

Our aim is to prove to clients that for a given block header $h_i$, $\mathsf{valid}(T_i, S_{i-i})$ returns $\mathsf{false}$ in less than $O(n)$ time and less than $O(n)$ space, relying on as few security assumptions as possible.

\subsection{Network Model}

We assume a network that consists of two types of nodes:

\begin{itemize}
    \item \textbf{Full nodes.} These are nodes which download and verify the entire blockchain. Honest full nodes store and rebroadcast valid blocks that they download to other full nodes, and broadcast block headers associated with valid blocks to light clients. Some of these nodes may participate in consensus (\ie by producing blocks).
    \item \textbf{Light clients.} These are nodes with computational capacity and network bandwidth that is too low to download and verify the entire blockchain. They receive block headers from full nodes, and on request, Merkle proofs that some transaction or state is a part of the block header.
\end{itemize}

We assume a network topology as shown in \Cref{fig:network-model}; full nodes communicate with each other, and light clients communicate with full nodes, but light clients do not communicate with each other. Additionally, we assume a maximum network delay $\delta$; such that if one honest node can connect to the network and download some data (\eg a block) at time $T$, then it is guaranteed that any other honest node will be able to do the same at time $T' \le T + \delta$.

\begin{figure*}
    \centering
    \includegraphics[width=0.6\linewidth,keepaspectratio]{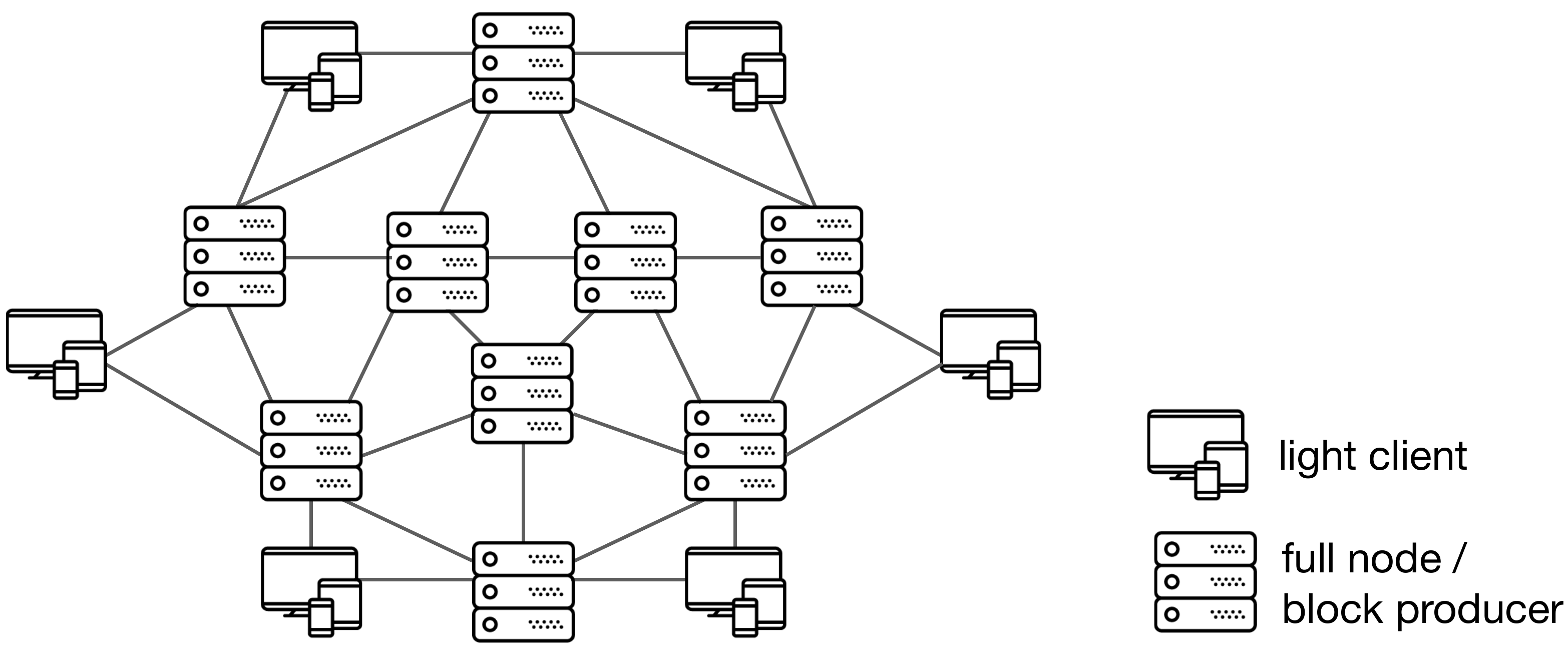}
    \caption{Network model---full nodes communicate with each other, and light clients communicate only with full nodes.}
    \label{fig:network-model}
\end{figure*}

\subsection{Threat Model}

We make the following assumptions in our threat model:

\begin{itemize}
    \item \textbf{Blocks and consensus.} Block headers may be created by adversarial actors, and thus may be invalid, and there is no honest majority of consensus-participating nodes that we can rely on.
    \item \textbf{Full nodes.} Full nodes may be dishonest, \eg they may not relay information (\eg fraud proofs), or they may relay invalid blocks. However, we assume that there is at least one honest full node that is connected to the network (\ie it is online, willing to generate and distribute fraud proofs, and is not under an eclipse attack \cite{heilman2015}).
    \item \textbf{Light clients.} We assume that each light client is connected to at least one honest full node. For data availability proofs, we assume a minimum number of honest light clients to allow for a block to be reconstructed. The specific number depends on the parameters of the system, and is analysed in \Cref{sec:sampling-security-analysis}.
\end{itemize}

Note that our goal is specifically to ensure that light clients do not accept blocks with invalid transactions, in the presence of a dishonest majority of consensus-participating nodes. This is different to double spending attacks, where a dishonest majority forks the chain to undo valid transactions, which is not the focus of this paper. An honest majority assumption is still necessary to prevent double spending attacks for both full nodes and light clients---our goal is to eliminate this assumption for transaction validity, thus significantly limiting the damage that a dishonest majority can do.

%% file: sections/fraud-proofs.tex
\section{Fraud Proofs}\label{sec:fraud-proofs}

\begin{figure*}
    \centering
    \includegraphics[width=0.5\linewidth,keepaspectratio]{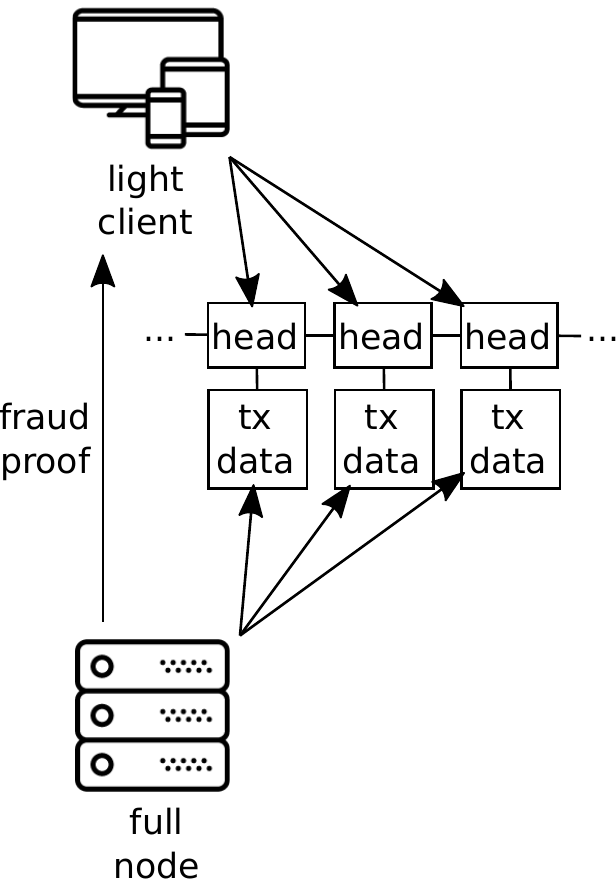}
    \caption{Overview of the architecture of a fraud proof system at a network level.}
    \label{fig:design}
\end{figure*}

\subsection{Block Structure}\label{sec:block-structure}

In order to support efficient fraud proofs, it is necessary to design a blockchain data structure that supports fraud proof generation by design. Extending the model described in \Cref{sec:blockchain-model}, a block header $h_i$ at height $i$ contains the following elements:

\begin{description}
    \item [$\mathsf{prevHash}_i$] The hash of the previous block header in the chain.
    \item [$\mathsf{dataRoot}_i$] The root of the Merkle tree of the data (\eg transactions) included in the block.
    \item [$\mathsf{dataLength}_i$] The number of leaves represented by $\mathsf{dataRoot}_i$.
    \item [$\mathsf{stateRoot}_i$] The root of a sparse Merkle tree of the state of the blockchain (to be described in \Cref{sec:state-roots}).
    \item [$\mathsf{additionalData}_i$] Additional arbitrary data that may be required by the network (\eg in proof-of-work, this may include a nonce and the target difficulty threshold).
\end{description}

Additionally, the hash of each block header $\mathsf{blockHash}_i = \mathsf{hash}(h_i)$ is also stored by clients and nodes.

Note that typically blockchains have the Merkle root of transactions included in headers. We have abstracted this to a `Merkle root of data' called $\mathsf{dataRoot}_i$, because as we shall see in \Cref{sec:shares-periods}, as well as including transactions in the block data, we also need to include intermediate state roots.

\subsection{State Root and Execution Trace Construction}\label{sec:state-roots}

To instantiate a blockchain based on the state-based model described in \Cref{sec:blockchain-model}, we make use of sparse Merkle trees, and represent the state as a key-value map. We explain how both a UTXO-based and an account-based blockchain can be instantiated atop such a model:
\begin{itemize}
    \item \textbf{UTXO-based.} The keys in the map are transaction output identifiers \eg $\mathsf{hash}(\mathsf{hash}(d)||i)$ where $d$ is the data of the transaction and $i$ is the index of the output being referred to in $d$. The value of each key is the state of each transaction output identifier: either $unspent$ ($1$) or $nonexistent$ ($0$, the default value).
    \item \textbf{Account-based.} This is already a key-value map, where the key is the account or storage variable, and the value is the balance of the account or the value of the variable.
\end{itemize}

The state would need to keep track of all data that is relevant to block processing, including for example the cumulative transaction fees paid to the creator of the current block after each transaction.

We now define a variation of the function $\mathsf{transition}$ defined in \Cref{sec:blockchain-model}, called $\mathsf{rootTransition}$, that performs transitions without requiring the whole state tree, but only the state root and Merkle proofs of parts of the state tree that the transaction reads or modifies (which we call ``state witness'', or $w$ for short). These Merkle proofs are effectively expressed as a sub-tree of the same state tree with a common root.
\begin{equation*}
\mathsf{rootTransition}(\mathsf{stateRoot}, t, w) \in \{\mathsf{stateRoot'}, \mathsf{err}\}
\end{equation*}

A state witness $w$ consists of a set of key-value pairs and their associated Sparse Merkle proofs in the state tree, $w = \{(k_1, v_1, \{k_1, v_1 \rightarrow \mathsf{stateRoot}\}), (k_2, v_2,\allowbreak \{k_2, v_2\allowbreak \rightarrow \mathsf{stateRoot}\}), ..., (k_n, v_n, \{k_n, v_n \rightarrow \mathsf{stateRoot}\})\}$.

After executing $t$ on the parts of the state shown by $w$, if $t$ modifies any of the state, then the new resulting $\mathsf{stateRoot}'$ can be generated by computing the root of the new sub-tree with the modified leafs. Note that if $w$ is invalid and does not contain all of the parts of the state required by $t$ during execution, then $\mathsf{err}$ is returned.

Let us denote, for the list of transactions $T_i = (t_i^0, t_i^1, ..., t_i^n)$, where $t_i^j$ denotes a transaction $j$ at block $i$, then $w_i^j$ is the state witness for transaction $w_i^j$ for $\mathsf{stateRoot}_i$.

Thus given the intermediate state roots after applying every transaction one at a time, $\mathsf{interRoot}_i^j = \mathsf{rootTransition}(\mathsf{interRoot}_i^{j-1}, t_i^j, w_i^j)$, and the base case $\mathsf{interRoot}_i^{-1} = \mathsf{stateRoot}_{i-1}$, then $\mathsf{stateRoot}_i = \mathsf{interRoot}_i^n$. Hence, $\mathsf{interRoot}_i^j$ denotes the intermediate state root at block $i$ after applying transactions $t_i^0, t_i^1, ..., t_i^j$.

\subsection{Data Root and Periods}\label{sec:shares-periods}

\begin{figure*}
    \centering
    \includegraphics[width=1\linewidth,keepaspectratio]{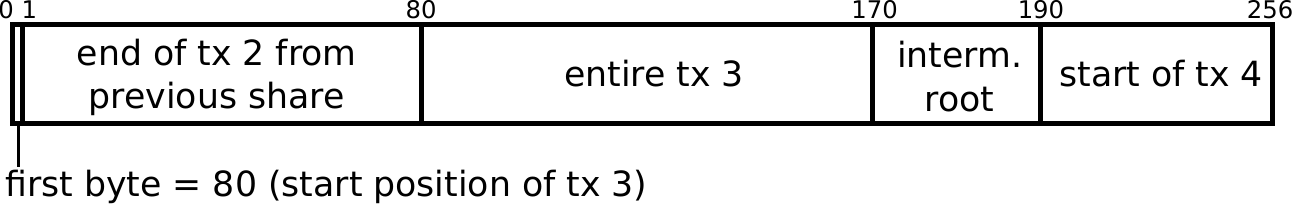}
    \caption{Example of a 256-byte share.}
    \label{fig:share}
\end{figure*}

The data represented by the $\mathsf{dataRoot}_i$ of a block contains transactions arranged into fixed-size chunks of data called `shares', interspersed with intermediate state roots called `traces' between transactions. We denote $\mathsf{trace}_i^j$ as the $j$th intermediate state root in block $i$. It is necessary to arrange data into fixed-size shares to allow for data availability proofs as we shall see in \Cref{sec:availability}. Each leaf in the data tree represents a share.

As a share may not contain entire transactions but only parts of transactions as shown in \Cref{fig:share}, we may reserve the first byte in each share to be the starting position of the first transaction that starts in the share, or $0$ if no transaction starts in the share. This allows a protocol message parser to establish the message boundaries without needing every transaction in the block.

Given a list of shares $(\mathsf{sh}_0, \mathsf{sh}_1, ..., \mathsf{sh}_n)$ we define a function $\mathsf{parseShares}$ which parses these shares and outputs an ordered list of $t$ messages $(m_0, m_1, ..., m_t)$, which are either transactions or intermediate state roots. For example, $\mathsf{parseShares}$ on some shares in the middle of some block $i$ may return $(\mathsf{trace}_i^1, t_i^4, t_i^5, t_i^6, \mathsf{trace}_i^2)$.
\begin{align*}
    \mathsf{parseShares}((\mathsf{sh}_0, \mathsf{sh}_1, ..., \mathsf{sh}_n)) = (m_0, m_1, ..., m_t)
\end{align*}

Note that as the block data does not necessarily contain an intermediate state root after every transaction, we assume a `period criterion', a protocol rule that defines how often an intermediate state root should be included in the block's data. For example, the rule could be at least once every $p$ transactions, or $b$ bytes or $g$ gas (\ie in Ethereum \cite{wood2018}).

We thus define a function $\mathsf{parsePeriod}$ which parses a list of messages, and returns a pre-state intermediate root $\mathsf{trace}_i^x$, a post-state intermediate root $\mathsf{trace}_i^{x+1}$, and a list of transaction $(t_i^g, t_i^{g+1}, ..., t_i^{g+h})$ such that applying these transactions on $\mathsf{trace}_i^x$ is expected to return $\mathsf{trace}_i^{x+1}$. If the list of messages violate the period criterion, then the function may return $\mathsf{err}$, for example if there are too many transactions in the messages to constitute a period.

\begin{align*}
    \mathsf{parsePeriod}((m_0, m_1, ..., m_t)) \in \{(\mathsf{trace}_i^x, \mathsf{trace}_i^{x+1}, (t_i^g, t_i^{g+1}, ..., t_i^{g+h})), \mathsf{err}\}
\end{align*}

Note that $\mathsf{trace}_i^x$ may be $nil$ if no pre-state root was parsed, as this may be the case if the first messages in the block are being parsed, and thus the pre-state root is the state root of the previous block $\mathsf{stateRoot}_{i-i}$. Likewise, $\mathsf{trace}_i^{x+1}$ may be $nil$ if no post-state root was parsed \ie if the last messages in the block are being parsed, as the post-state root would be $\mathsf{stateRoot}_i$.

\subsection{Proof of Invalid State Transition}\label{sec:proof-invalid-state-transition}

A faulty or malicious miner may provide an incorrect $\mathsf{stateRoot}_i$. We can use the execution trace provided in $\mathsf{dataRoot}_i$ to prove that some part of the execution trace was invalid.

We define a function $\mathsf{VerifyTransitionFraudProof}$ and its parameters which verifies fraud proofs received from full nodes. We denote $d_i^j$ as share number $j$ in block $i$.

\textbf{Summary of \textsf{VerifyTransitionFraudProof}.} A fraud proof consists of the relevant shares in the block that contain a bad state transition, Merkle proofs for those shares, and the state witnesses for the transactions contained in those shares. The function takes as input a fraud proof, and checks if applying the transactions in a period of the block's data on the intermediate pre-state root results in the intermediate post-state root specified in the block data. If it does not, then the fraud proof is valid, and the block that the fraud proof is for should be permanently rejected by the client.

\begin{align*}
&\mathsf{VerifyTransitionFraudProof}(\mathsf{blockHash}_i,\\
    &\indent(d_i^y, d_i^{y+1}, ..., d_i^{y+m}), y, \tag*{(shares)}\\
    &\indent(\{d_i^y \rightarrow \mathsf{dataRoot}_i\}, \{d_i^{y+1} \rightarrow \mathsf{dataRoot}_i\}, ..., \{d_i^{y+m} \rightarrow \mathsf{dataRoot}_i\}),\\
    &\indent(w_i^y, w_i^{y+1}, ..., w_i^{y+m}), \tag*{(state witnesses)}\\
&) \in \{\mathsf{true}, \mathsf{false}\}
\end{align*}

$\mathsf{VerifyTransitionFraudProof}$ returns $\mathsf{true}$ if all of the following conditions are met, otherwise $\mathsf{false}$ is returned:
\begin{enumerate}
    \item $\mathsf{blockHash}_i$ corresponds to a block header $h_i$ that the client has downloaded and stored.
    \item For each share $d_i^{y+a}$ in the proof, $\mathsf{VerifyMerkleProof}(d_i^{y+a}, \{d_i^{y+a} \rightarrow \mathsf{dataRoot}_i\},\allowbreak \mathsf{dataRoot}_i, \mathsf{dataLength}_i, y+a)$ returns $\mathsf{true}$.
    \item Given $\mathsf{parsePeriod}(\mathsf{parseShares}((d_i^y, d_i^{y+1}, ..., d_i^{y+m}))) \in \{(\mathsf{trace}_i^x, \mathsf{trace}_i^{x+1}, (t_i^g,\allowbreak t_i^{g+1}, ..., t_i^{g+h})), \mathsf{err}\}$, the result must not be $\mathsf{err}$. If $\mathsf{trace}_i^x$ is $nil$, then $y = 0$ is true, and if $\mathsf{trace}_i^{x+1}$ is $nil$, then $y+m = \mathsf{dataLength}_i$ is true.
    \item Check that applying $(t_i^g,\allowbreak t_i^{g+1}, ..., t_i^{g+h})$ on $\mathsf{trace}_i^x$ results in $\mathsf{trace}_i^{x+1}$. Formally, let the intermediate state roots after applying every transaction in the proof one at a time be $\mathsf{interRoot}_i^j = \mathsf{rootTransition}(\mathsf{interRoot}_i^{j-1}, t_i^j, w_i^j)$. If $\mathsf{trace}_i^x$ is not $nil$, then the base case is $\mathsf{interRoot}_i^{y} = \mathsf{trace}_i^x$, otherwise $\mathsf{interRoot}_i^{y} = \mathsf{stateRoot}_{i-1}$. If $\mathsf{trace}_i^{x+1}$ is not $nil$, $\mathsf{trace}_i^{x+1} = \mathsf{interRoot}_i^{g+h}$ is true, otherwise $\mathsf{stateRoot}_{i} = \mathsf{interRoot}_i^{y+m}$ is true.\footnote{For simplicity, we assume a model where state witnesses are provided for every individual intermediate state root within the trace, but it is also possible to only provide state witnesses only for the trace intermediate pre-state root, and execute the transactions as a single batch.}
\end{enumerate}

\subsection{Transaction Fees}

As discussed in \Cref{sec:state-roots}, the state would need to keep track of all data that is relevant to block processing. A block producer may attempt to collect more transaction fees than is afforded to them by the transactions in the block. In order to make this detectable by a fraud proof as part of the model we have described, we can introduce a special key in the state tree called $\mathsf{\_\_fees\_\_}$, which represents the cumulative fees in the block after applying each transaction, and is reset to $0$ after applying the transaction where the block producer collects the fees.

%% file: sections/availability.tex
\section{Data Availability Proofs}\label{sec:availability}

A malicious block producer could prevent full nodes from generating fraud proofs by withholding the data needed to recompute $\mathsf{dataRoot}_i$ and only releasing the block header to the network. The block producer could then only release the data---which may contain invalid transactions or state transitions---long after the block has been published, and make the block invalid. This would cause a rollback of transactions on the ledger of future blocks. It is therefore necessary for light clients to have a level of assurance that the data matching $\mathsf{dataRoot}_i$ is indeed available to the network.

We propose a data availability scheme based on Reed-Solomon erasure coding, where light clients request random shares of data to get high probability guarantees that all the data associated with the root of a Merkle tree is available. The scheme assumes there is a sufficient number of honest light clients making the same requests such that the network can recover the data, as light clients upload these shares to full nodes, if a full node who does not have the complete data requests it. It is fundamental for light clients to have assurance that all the transaction data is available, because it is only necessary to withhold a few bytes to hide an invalid transaction in a block.

We define below \emph{soundness} and \emph{agreement} and analyse them in \Cref{sec:properties-security-analysis}.
\begin{definition}[Soundness]\label{def:soundness}
    If an honest light client accepts a block as available, then at least one honest full node has the full block data or will have the full block data within some known maximum delay $k * \delta$ where $\delta$ is the maximum network delay.
\end{definition}
\begin{definition}[Agreement]\label{def:agreement}
    If an honest light client accepts a block as available, then all other honest light clients will accept that block as available within some known maximum delay $k * \delta$ where $\delta$ is the maximum network delay.
\end{definition}

\subsection{Strawman 1D Reed-Solomon Availability Scheme}\label{sec:strawman-availability}

To provide some intuition, we first describe a strawman data availability scheme, based on standard Reed-Solomon coding.

A block producer compiles a block of data consisting of $k$ shares, extends the data to $2k$ shares using Reed-Solomon encoding, and computes a Merkle root (the $\mathsf{dataRoot}_i$) over the extended data, where each leaf corresponds to one share.

When light clients receive a block header with this $\mathsf{dataRoot}_i$, they randomly sample shares from the Merkle tree that $\mathsf{dataRoot}_i$ represents, and only accept a block once it has received all of the shares requested. If an adversarial block producer makes more than 50\% of the shares unavailable to make the full data unrecoverable (recall in \Cref{sec:erasure-codes} that Reed-Solomon codes allow recovery of $2t$ shares from any $t$ shares), there is a 50\% chance that a client will randomly sample an unavailable share in the first draw, a 25\% chance after two draws, a 12.5\% chance after three draws, and so on, if they draw with replacement. (In the full scheme, they will draw without replacement, and so the probability will be even lower.)

Note that for this scheme to work, there must be enough light clients in the network sampling enough shares so that block producers will be required to release more than 50\% of the shares in order to pass the sampling challenge of all light clients, and so that the full block can be recovered. An in-depth probability and security analysis is provided in \Cref{sec:sampling-security-analysis}.

The problem with this scheme is that an adversarial block producer may incorrectly construct the extended data, and thus the incomplete block is unrecoverable from the extended data even if more than 50\% of the data is available. With standard Reed-Solomon encoding, the fraud proof that the extended data is invalid is the original data itself, as clients would have to re-encode all data locally to verify the mismatch with the given extended data, and thus it requires $O(n)$ data with respect to the size of the block. Therefore, we instead use multi-dimensional encoding, as described in \Cref{sec:2d-rs-tree}, so that proofs of incorrectly generated codes are limited to a specific axis---rather than the entire data---reducing proof size to $O(\sqrt[d]{n})$ where $d$ is the number of dimensions of the encoding. For simplicity, we will only consider two-dimensional Reed-Solomon encoding in this paper, but our scheme can be generalised to higher dimensions.

We note in \Cref{sec:succinct-proofs} that succinct proofs of computation could be an alternative future solution to this problem instead of multi-dimensional encoding.

\subsection{2D Reed-Solomon Encoded Merkle Tree Construction}\label{sec:2d-rs-tree}

\begin{figure*}
    \centering
    \includegraphics[width=0.55\linewidth,keepaspectratio]{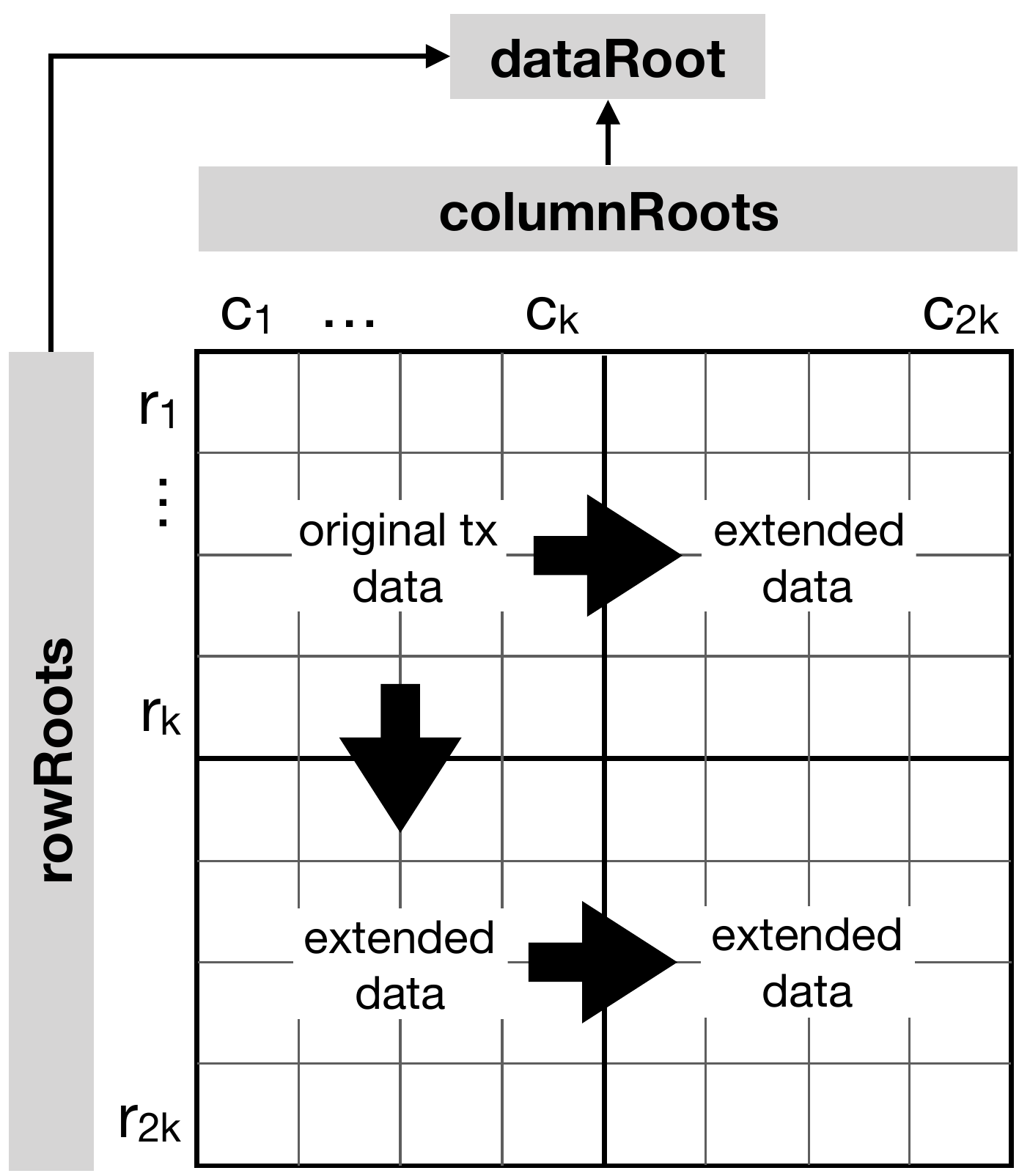}
    \caption{Diagram showing a 2D Reed-Solomon encoding. The original data is initially arranged in a  $k \times k$ matrix, which is then `extended' to a  $2k \times 2k$ matrix applying multiple times Reed-Solomon encoding.}
   \label{fig:reedsolomon}
\end{figure*}

A 2D Reed-Solomon Encoded Merkle tree can be constructed as follows from a block of data:
\begin{enumerate}
    \item Split the raw data into shares of size $\mathsf{shareSize}$ each, and arrange them into a $k \times k$ matrix; apply padding if the last share is not exactly of size $\mathsf{shareSize}$, or if there are not enough shares to complete the matrix.
    \item Apply Reed-Solomon encoding on each row and column of the $k \times k$ matrix to extend the data horizontally and vertically; \ie encode each row and each column. Then apply a third time a Reed-Solomon encoding horizontally, on the vertically extended portion of the matrix to create a $2k \times 2k$ matrix, as shown in~\Cref{fig:reedsolomon}. This results in an extended matrix $M_i$ for block $i$.
    \item Compute the root of the Merkle tree for each row and column in the $2k \times 2k$ matrix, where each leaf is a share. We have $\mathsf{rowRoot}_i^j = \mathsf{root}((M_i^{j,1}, M_i^{j,2},\allowbreak ..., M_i^{j,2k}))$ and $\mathsf{columnRoot}_i^j = \mathsf{root}((M_i^{1,j}, M_i^{2,j}, ..., M_i^{2k,j}))$, where $M_i^{x,y}$ represents the share in row $x$, column $y$ in the matrix.
    \item Compute the root of the Merkle tree of the roots computed in step 3 and use this as $\mathsf{dataRoot}_i$. We have $\mathsf{dataRoot}_i = \mathsf{root}((\mathsf{rowRoot}_i^1, \mathsf{rowRoot}_i^2, ...,\allowbreak \mathsf{rowRoot}_i^{2k}, \mathsf{columnRoot}_i^1, \mathsf{columnRoot}_i^2, ..., \mathsf{columnRoot}_i^{2k}))$.
\end{enumerate}

The resulting tree of $\mathsf{dataRoot}_i$ has $\mathsf{dataLength}_i = 2 \times (2k)^2$ elements, where the first $\frac{1}{2}\mathsf{dataLength}_i$ elements are in leaves via the row roots, and the latter half are in leaves via the column roots.

Note that although it is possible to present a Merkle proof from $\mathsf{dataRoot}_i$ to an individual share, it is important to note that a Merkle tree has $2^x$ leaves, and the Merkle sub-trees for the row and column roots are constructed independently from $\mathsf{dataRoot}_i$. Therefore it is necessary to have a wrapper function around $\mathsf{VerifyMerkleProof}$ called $\mathsf{VerifyShareMerkleProof}$ with the same parameters which takes into account how the underlying Merkle tree deals with an unbalanced number of leaves; this may involve calling $\mathsf{VerifyMerkleProof}$ twice for different portions of the path, or offsetting the index.\footnote{For example, if the underlying tree simply repeats the last leaves to pad the tree to $2^x$ leaves, then the wrapper function may be $\mathsf{VerifyShareMerkleProof}(e, \{e \rightarrow r\}, r, n, i) = \mathsf{VerifyMerkleProof}(e, \{e \rightarrow r\}, r, n, i + \lfloor i \mathbin{/} \sqrt{\frac{1}{2}i}\rfloor \times (2^{\lceil\log_{2}(\sqrt{\frac{1}{2}i})\rceil}-\sqrt{\frac{1}{2}i}))$.}

The width of the matrix can be derived as $\mathsf{matrixWidth}_i = \sqrt{\frac{1}{2}\mathsf{dataLength}_i}$. If we are only interested in the row and column roots of $\mathsf{dataRoot}_i$, rather than the actual shares, then we can assume that $\mathsf{dataRoot}_i$ has $2 \times \mathsf{matrixWidth}_i$ leaves when verifying a Merkle proof of a row or column root.

A light client or full node is able to reconstruct $\mathsf{dataRoot}_i$ from all the row and column roots by recomputing step 4. In order to gain data availability assurances, all light clients should at minimum download all the row and column roots needed to reconstruct $\mathsf{dataRoot}_i$ and check that step 4 was computed correctly, because as we shall see in \Cref{sec:fraud-proofs-extended-data}, they are necessary to generate fraud proofs of incorrectly generated extended data.

We nevertheless represent all of the row and column roots as a a single $\mathsf{dataRoot}_i$ to allow `super-light' clients which do not download the row and column roots, but these clients cannot be assured of data availability and thus do not fully benefit from the increased security of allowing fraud proofs.

\subsection{Random Sampling and Network Block Recovery}\label{sec:random-sampling}

In order for any share in the 2D Reed-Solomon matrix to be unrecoverable, then at least $(k+1)^2$ out of $(2k)^2$ shares must be unavailable (see \Cref{th:min-shares-unavailable}). Thus when light clients receive a new block header from the network, they should randomly sample $0 < s < (k+1)^2$ distinct shares from the extended matrix, and only accept the block if they receive all shares. Additionally, light clients gossip shares that they have received to the network, so that the full block can be recovered by honest full nodes.

The protocol between a light client and the full nodes that it is connected to works as follows:
\begin{enumerate}
    \item The light client receives a new block header $h_i$ from one of the full nodes it is connected to, and a set of row and column roots $R = (\mathsf{rowRoot}_i^1, \mathsf{rowRoot}_i^2, ...,\allowbreak \mathsf{rowRoot}_i^{2k}, \mathsf{columnRoot}_i^1, \mathsf{columnRoot}_i^2, ..., \mathsf{columnRoot}_i^{2k})$. If the check $\mathsf{root}(R)\allowbreak = \mathsf{dataRoot}_i$ is false, then the light client rejects the header.
    \item The light client randomly chooses a set of unique $(x, y)$ coordinates $S = \{(x_0, y_0) (x_1, y_1), ..., (x_n, y_n)\}$ where $0 < x \leq \mathsf{matrixWidth}_i$ and $0 < y \leq \mathsf{matrixWidth}_i$, corresponding to points on the extended matrix, and sends them to one or more of the full nodes it is connected to.
    \item If a full node has all of the shares corresponding to the coordinates in $S$ and their associated Merkle proofs, then for each coordinate $(x_a, y_b)$ the full node responds with $M_i^{x_a,y_b}, \{M_i^{x_a,y_b} \rightarrow \mathsf{rowRoot}_i^a\}$ or $M_i^{x_a,y_b}, \{M_i^{x_a,y_b} \rightarrow \mathsf{columnRoot}_i^b\}$. Note that there are two possible Merkle proofs for each share; one from the row roots, and one from the column roots, and thus the full node must also specify for each Merkle proof if it is associated with a row or column root.
    \item For each share $M_i^{x_a,y_b}$ that the light client has received, the light client checks $\mathsf{VerifyMerkleProof}(M_i^{x_a,y_b}, \{M_i^{x_a,y_b} \rightarrow \mathsf{rowRoot}_i^a\}, \mathsf{rowRoot}_i^a, \mathsf{matrixWidth}_i, b)$ is $\mathsf{true}$ if the proof is from a row root, otherwise if the proof is from a column root then $\mathsf{VerifyMerkleProof}(M_i^{x_a,y_b}, \{M_i^{x_a,y_b} \rightarrow \mathsf{columnRoot}_i^b\}, \mathsf{columnRoot}_i^b,\allowbreak \mathsf{matrixWidth}_i, a)$ is $\mathsf{true}$.
    \item Each share and valid Merkle proof that is received by the light client is gossiped to all the full nodes that the light client is connected to if the full nodes do not have them, and those full nodes gossip it to all of the full nodes that they are connected to.
    \item If all the proofs in step 4 succeeded, and no shares are missing from the sample made in step 2, then the block is accepted as available if within $2 \times \delta$ no fraud proofs for the block's erasure code is received (\Cref{sec:fraud-proofs-extended-data}).
\end{enumerate}

\subsection{Selective Share Disclosure}\label{sec:selective-share-disclosure}

If a block producer selectively releases shares as light clients ask for them, up to $(k+1)^2$ shares, they can violate the soundness property (\Cref{def:soundness}) of the clients that ask for the first $(k+1)^2$ out of $(2k)^2$ shares, as they will accept the blocks as available despite them being unrecoverable.

This can be alleviated if one assumes an enhanced network model where a sufficient number of honest light clients make requests such that more than $(k+1)^2$ shares will be sampled, and that each sample request for each share is anonymous (\ie sample requests cannot be linked to the same client) and the distribution in which every sample request is received is uniformly random, for example by using a mix net \cite{chaum1981}. As the network would not be able to link different per-share sample requests to the same clients, shares cannot be selectively released on a per-client basis.

We thus assume two network connection models that sample requests can be made under, which we will analyse in the security analysis:
\begin{itemize}
    \item \textbf{Standard model.} Sample requests are linkable to the clients that made them, and the order that they are received is predictable (\eg they are received in the order that they were sent).
    \item \textbf{Enhanced model.} Different sample requests cannot be linked to the same client, and the order that they are received by the network is uniformly random with respect to other requests.
\end{itemize}

\subsection{Fraud Proofs of Incorrectly Generated Extended Data}\label{sec:fraud-proofs-extended-data}

If a full node has enough shares to recover a particular row or column, and after doing so detects that recovered data does not match its respective row or column root, then it must distribute a fraud proof consisting of enough shares in that row or column to be able to recover it, and a Merkle proof for each share.

We define a function $\mathsf{VerifyCodecFraudProof}$ that verifies these fraud proofs, where $\mathsf{axisRoot}_i^j \in \{\mathsf{rowRoot}_i^j, \mathsf{columnRoot}_i^j\}$. These proofs can also be verified by `super-light' clients as they do not assume any knowledge of the row and column roots. We denote $\mathsf{axis}$ and $\mathsf{ax}_j$ as row or column boolean indicators; $0$ for rows and $1$ for columns.

\textbf{Summary of \textsf{VerifyCodecFraudProof}.} The fraud proof consists of the Merkle root of the incorrectly generated row or column, a Merkle proof that the root is in the data tree, enough shares to be able to reconstruct that row or column, and a Merkle proof that each share is in the data tree. The function takes as input a fraud proof, and checks that \first all of the shares given by the prover are in the same row or column and \second that the recovered row or column does not match the row or column root in the block. If both conditions are true, then the fraud proof is valid, and the block that the fraud proof is for should be permanently rejected by the client.

\begin{align*}
&\mathsf{VerifyCodecFraudProof}(\mathsf{blockHash}_i,\\
    &\indent \mathsf{axisRoot}_i^j, \{\mathsf{axisRoot}_i^j \rightarrow \mathsf{dataRoot}_i\}, j, \tag*{(row or column root)}\\
    &\indent \mathsf{axis}, \tag*{(row or column indicator)}\\
    &\indent((\mathsf{sh}_0, \mathsf{pos}_0, \mathsf{ax}_0), (\mathsf{sh}_1, \mathsf{pos}_1, \mathsf{ax}_1), ..., (\mathsf{sh}_k, \mathsf{pos}_k, \mathsf{ax}_k)), \tag*{(shares)}\\
    &\indent(\{\mathsf{sh}_0 \rightarrow \mathsf{dataRoot}_i\}, \{\mathsf{sh}_1 \rightarrow \mathsf{dataRoot}_i\}, ..., \{\mathsf{sh}_k \rightarrow \mathsf{dataRoot}_i\})\\
&) \in \{\mathsf{true}, \mathsf{false}\}
\end{align*}

Let $\mathsf{recover}$ be a function that takes a list of shares and their positions in the row or column $((\mathsf{sh}_0, \mathsf{pos}_0), (\mathsf{sh}_1, \mathsf{pos}_1), ..., (\mathsf{sh}_k, \mathsf{pos}_k))$, and the length of the original row or column $k$. The function outputs the full recovered shares $(\mathsf{sh}_0, \mathsf{sh}_1, ..., \mathsf{sh}_{2k})$ or $\mathsf{err}$ if the shares are unrecoverable.
\begin{equation*}
    \mathsf{recover}(((\mathsf{sh}_0, \mathsf{pos}_0), (\mathsf{sh}_1, \mathsf{pos}_1), ..., (\mathsf{sh}_k, \mathsf{pos}_k)), k) \in \{(\mathsf{sh}_0, \mathsf{sh}_1, ..., \mathsf{sh}_{2k}), \mathsf{err}\}
\end{equation*}

$\mathsf{VerifyCodecFraudProof}$ returns true if all of the following conditions are met:
\begin{enumerate}
    \item $\mathsf{blockHash}_i$ corresponds to a block header $h_i$ that the client has downloaded and stored.
    \item If $\mathsf{axis} = 0$ (row root), $\mathsf{VerifyMerkleProof}(\mathsf{axisRoot}_i^j, \{\mathsf{axisRoot}_i^j \rightarrow \mathsf{dataRoot}_i\},\allowbreak \mathsf{dataRoot}_i, 2 \times \mathsf{matrixWidth}_i, j)$ returns $\mathsf{true}$.
    \item If $\mathsf{axis} = 1$ (col. root), $\mathsf{VerifyMerkleProof}(\mathsf{axisRoot}_i^j, \{\mathsf{axisRoot}_i^j \rightarrow \mathsf{dataRoot}_i\},\allowbreak \mathsf{dataRoot}_i, 2 \times \mathsf{matrixWidth}_i, \frac{1}{2}\mathsf{dataLength}_i + j)$ returns $\mathsf{true}$.
    \item For each $(\mathsf{sh}_x, \mathsf{pos}_x, \mathsf{ax}_x)$, $\mathsf{VerifyShareMerkleProof}(\mathsf{sh}_x, \{\mathsf{sh}_x \rightarrow \mathsf{dataRoot}_i\},\allowbreak \mathsf{dataRoot}_i, \mathsf{dataLength}_i, \mathsf{index})$ returns true, where $\mathsf{index}$ is the expected index of the $\mathsf{sh}_x$ in the data tree based on $\mathsf{pos}_x$ assuming it is in the same row or column as $\mathsf{axisRoot}_i^j$. See \Cref{sec:computation-index} for how $\mathsf{index}$ can be computed.

    Note that full nodes can specify Merkle proofs of shares in rows or columns from either the row or column roots \eg if a row is invalid but the full nodes only has Merkle proofs for the row's share from column roots. This also allows for full nodes to generate fraud proofs if there are inconsistencies in the data between rows and columns \eg if the same cell in the matrix has a different share in its row and column trees.
    \item $\mathsf{root}(\mathsf{recover}(((\mathsf{sh}_0, \mathsf{pos}_0), (\mathsf{sh}_1, \mathsf{pos}_1), ..., (\mathsf{sh}_k, \mathsf{pos}_k)), k)) = \mathsf{axisRoot}_i^j$ is false.
\end{enumerate}

If $\mathsf{VerifyCodecFraudProof}$ for $\mathsf{blockHash}_i$ returns $\mathsf{true}$, then the block header $h_i$ is permanently rejected by the light client.

\subsection{Sampling Security Analysis}\label{sec:sampling-security-analysis}

We present how the data availability scheme presented in \Cref{sec:availability} can provide lights clients with a high level of assurance that block data is available to the network.

\subsubsection{Minimum Unavailable Shares for Unrecoverability}

\Cref{th:min-shares-unavailable} states that data is unrecoverable if a malicious block proposer withholds $k+1$ shares of at least $k+1$ columns or rows; which makes a total of $(k+1)^2$ shares to withhold.

\begin{figure*}
    \centering
    \includegraphics[width=.5\linewidth,keepaspectratio]{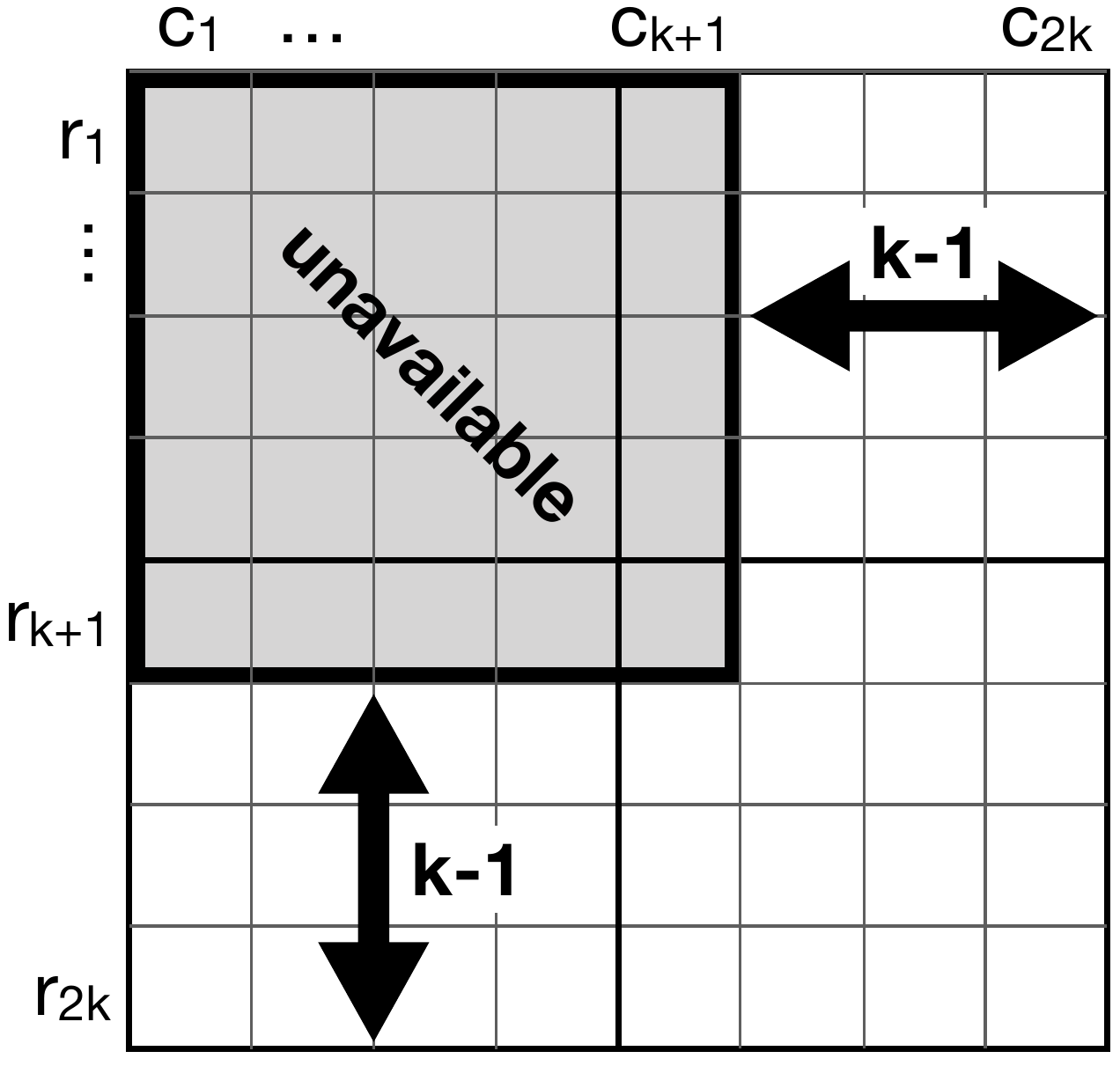}
    \caption{Graphical interpretation of \Cref{th:min-shares-unavailable}. Data is unrecoverable if at least $k+1$ columns (or rows) have each at least $k+1$ unavailable shares.}
   \label{fig:unavailable}
\end{figure*}

\begin{theorem}\label{th:min-shares-unavailable}
    Given a $2k \times 2k$ matrix $E$ as show in \Cref{fig:reedsolomon}, data is unrecoverable if at least $k+1$ columns or rows have each at least $k+1$ unavailable shares. In that case, the minimum number of shares that must be unrecoverable is $(k+1)^2$.
\end{theorem}
\begin{proof}
    Suppose a malicious block producer wants to make unrecoverable a share $E_{i,j}$ of the $2k \times 2k$ matrix $E$. Recall that Reed-Solomon encoding allow to recover all $2k$ shares from any $k$ shares; the block producer will have to \first make unrecoverable at least $k+1$ shares from the row $E_{i,*}$, and \second make unrecoverable at least $k+1$ shares from the column $E_{*,j}$.

    Let us start from \first; the block producer withholds at least $k+1$ shares from row $E_{i,*}$. However, each of these $k+1$ withheld shares $(E_{i,c_1}, \dots, E_{i,c_{k+1}}) \in E_{i,*}$ can be recovered from the available shares of their respective columns $E_{*,c_1}$, $E_{*,c_2}$\dots, $E_{*,c_{k+1}}$. Therefore, the block producer will also have to withhold at least $k+1$ shares from each of these columns. This gives a total of $(k+1)*(k+1)=(k+1)^2$ shares to withhold. Note that at this point, there are not enough shares left in the matrix to recover any of the $(k+1)^2$ shares of columns $(E_{*,c_1}, \dots E_{*,c_{k+1}})$.

    Let us now consider \second; the block producer withholds at least $k+1$ shares from the column $E_{*,j}$ to make unrecoverable the share $E_{i,j}$. As before, each shares $(E_{r_1,j}, \dots, E_{r_{k+1},j}) \in E_{*,j}$ can be recovered from the available shares of their respective row $E_{r_1,*}, E_{r_2,*}, \dots, E_{r_{k+1},*}$. Therefore, the block producer will also have to withhold at least at least $k+1$ shares from each of these rows. As before, this also gives a total of $(k+1)*(k+1)=(k+1)^2$ shares to withhold.

    However, \first is equivalent to \second by the symmetry of the matrix, and are actually operating on the same shares; executing \first on matrix $E$ is equivalent to executing \second on the transpose of matrix $E$.
\end{proof}\mustafa{@Alberto: Double check this - I fixed a typo (`transposed of the matrix E' $\rightarrow$ `transpose of matrix E').}

\subsubsection{Unrecoverable Block Detection}

\Cref{th:1-player} states the probability that a single light client will sample at least one unavailable share in a matrix with the minimum unavailable shares for unrecoverability, thus detecting that a block may be unrecoverable.

\begin{theorem}\label{th:1-player}
    Given a $2k \times 2k$ matrix $E$ as shown in \Cref{fig:reedsolomon}, where $(k+1)^2$ shares are unavailable. If one player randomly samples $0 < s < (k+1)^2$ shares from $E$, the probability of sampling at least one unavailable share is:
    \begin{equation}\label{eq:p1}
        p_1(X \geq 1) = 1 - \prod_{i=0}^{s-1}\left( 1-\frac{(k+1)^2}{4k^2-i} \right)
    \end{equation}
\end{theorem}
\begin{proof}
    We start by assuming that the $2k \times 2k$ matrix $E$ contains $q$ unavailable shares; If the player performs $m$ trials ($0 < s < (k+1)^2$), the probability of finding exactly zero unavailable shares is:
    \begin{equation}\label{eq:p11}
        p_{1}(X=0) = \frac{\binom{4k^2-q}{s}}{\binom{4k^2}{s}}
    \end{equation}
    The numerator of \Cref{eq:p11} computes the number of ways to pick $s$ chunks among the set of unavailable shares $4k^2-q$ (i.e., $\binom{4k^2-q}{s}$). The denominator computes the total number of ways to pick any $s$ samples out of the total number of samples (i.e., $\binom{4k^2}{s}$).

    Then, the probability  $p_{1}(X\geq1)$ of finding at least one unavailable share can be easily computed from \Cref{eq:p11}:
    \begin{eqnarray}\label{eq:p1_long}
        p_{1}(X\geq1) &=& 1-p_{1}(X=0)\\
        &=& 1 - \frac{\binom{4k^2-q}{s}}{\binom{4k^2}{s}}\\
        &=& 1 - \prod_{i=0}^{s-1}\left( 1-\frac{q}{4k^2-i}\right)
    \end{eqnarray}
    which can be re-written as \Cref{eq:p1} by setting $q=(k+1)^2$.
\end{proof}

\begin{figure*}
    \centering
    \includegraphics[width=\linewidth,keepaspectratio]{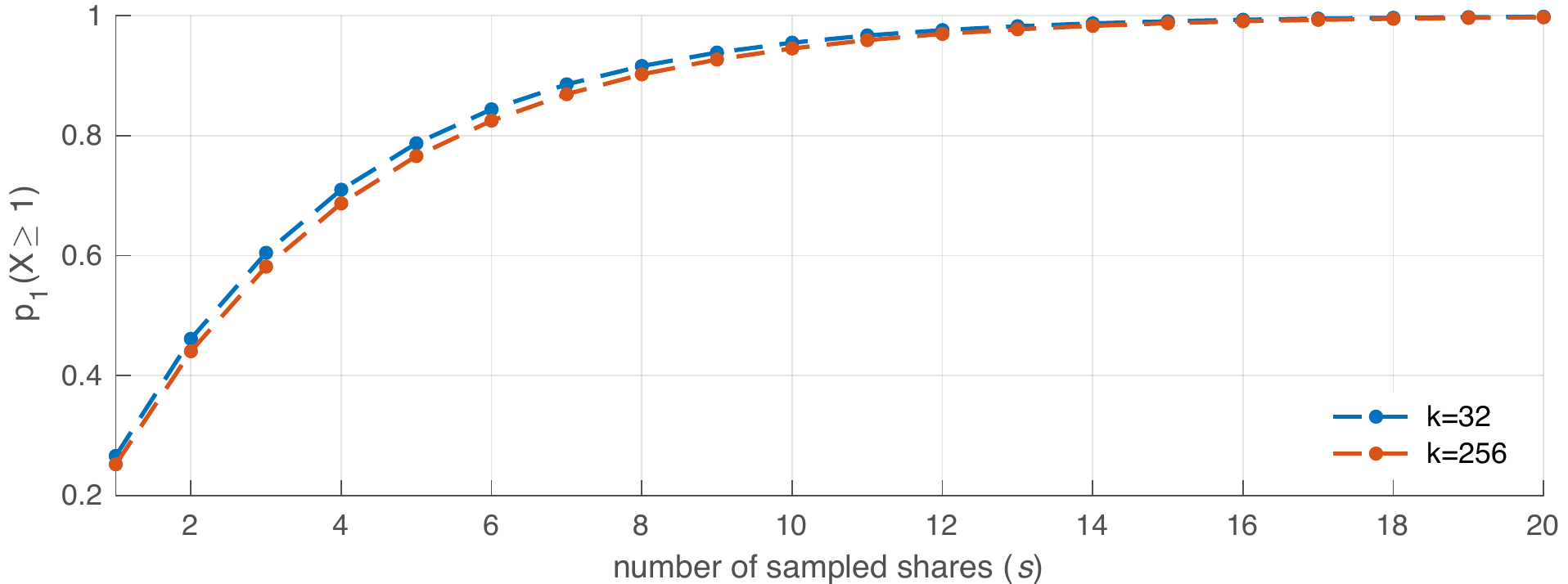}
    \caption{Plot of \Cref{eq:p1}---variation of the probability $p_1(X \geq 1)$ with the number of sampled shares ($s$) (computed for $k=32$ and $k=256$).}
   \label{fig:p1_s}
\end{figure*}

\begin{figure*}
    \centering
    \includegraphics[width=\linewidth,keepaspectratio]{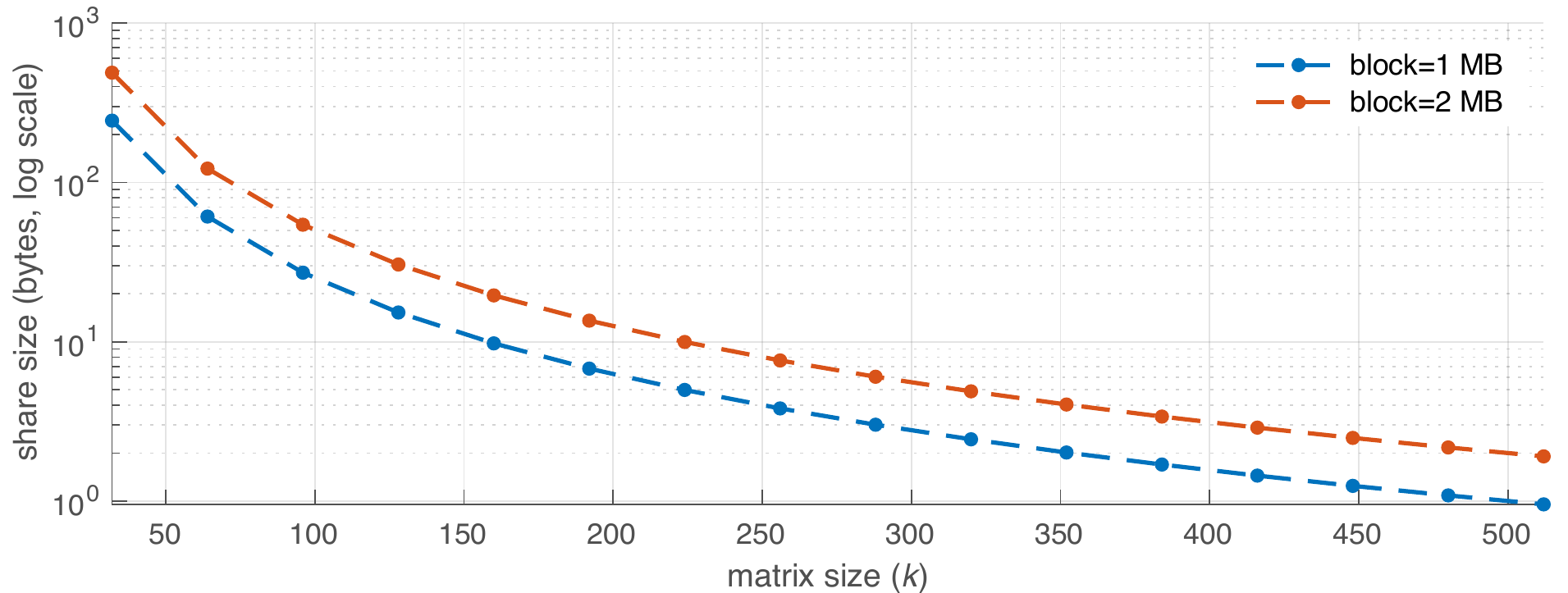}
    \caption{Variation of the shares size with the size of the matrix ($k$).}
   \label{fig:share_size}
\end{figure*}

\Cref{fig:p1_s} shows how this probability varies with $s$ samples for $k=32$ and $k=256$; each light client samples at least one unavailable share with about $60\%$ probability after 3 samplings (\ie after querying respectively $0.07\%$ of the block shares for $k=32$ and $0.001\%$ of the block shares for $k=256$), and with more than $99\%$ probability after 15 samplings (\ie after querying respectively $0.4\%$ of the block shares for $k=16$ and $0.005\%$ of the block shares for $k=256$). \Cref{fig:share_size} shows that light clients would have to download about 3.6 KB of shares to be able to detect incomplete blocks with more than $99\%$ probability for $k=32$, and about 57 bytes of shares for $k=256$.

\Cref{eq:p1_lim} shows a noticeable result: the probability $p_1(X \geq 1)$ is almost independent of $k$ for large values of $k$; it is therefore convenient to have a large matrix size (\ie $k \geq 128$) as this reduces the amount of data that light clients have to download.
\begin{equation}\label{eq:p1_lim}
    \lim_{k \to \infty} p_{1}(X\geq1) = \lim_{k \to \infty} \left(1 - \prod_{i=0}^{s-1}\left( 1-\frac{(k+1)^2}{4k^2-i} \right)\right) = 1 - \left(3/4\right)^s
\end{equation}

Under the enhanced model described in \Cref{sec:selective-share-disclosure}, a malicious block producer could statistically link light clients based on the shares they query; \ie assuming that a light client would never request twice the same share, a block producer can deduce that any request for the same share comes from a different client. To mitigate this problem, light clients could sample without replacement by performing the procedure for sampling with replacement multiple times, and only stop when they have sampled $s$ unique values.

\subsubsection{Multi-Client Unrecoverable Block Detection}

\Cref{th:c-player} captures the probability that more than $\hat{c}$ out of $c$ light clients sample at least one unavailable share in a matrix with the minimum unavailable shares for unrecoverability.

\begin{theorem}\label{th:c-player}
    Given a $2k \times 2k$ matrix $E$ as shown in \Cref{fig:reedsolomon}, where $(k+1)^2$ shares are unavailable. If $c$ players randomly sample $0 < s < (k+1)^2$ shares from $E$, the probability that more than $\hat{c}$ players sample at least one unavailable share is:
    \begin{equation}\label{eq:pc}
        p_{c}(Y > \hat{c}) = 1-\sum_{j=1}^{\hat{c}} \binom{c}{j} \bigr(p_1(X \geq 1)\bigr)^j \bigr(1-p_1(X \geq 1)\bigr)^{c-j}
    \end{equation}
    where $p_1(X \geq 1)$ is given by \Cref{eq:p1}.
\end{theorem}
\begin{proof}
    We start by computing the probability that exactly $\hat{c}$ players sample at least one unavailable share; this probability is given by the binomial probability mass function:
    \begin{equation}\label{eq:pcc}
        p_{s,\hat{c}}(Y = \hat{c}) =  \binom{c}{\hat{c}} \bigr(p_1(X \geq 1)\bigr)^{\hat{c}} \bigr(1-p_1(X \geq 1)\bigr)^{c-\hat{c}}
    \end{equation}
    where $p_1(X \geq 1)$ is given by \Cref{eq:p1}. \Cref{eq:pcc} describes the probability that $\hat{c}$ players succeed to sample at least one unavailable share. This can be viewed as the probability of observing $\hat{c}$ successes each happening with probability $p_1$, and $(c-\hat{c})$ failures each happening with probability $1-p_1$; there are $\binom{c}{\hat{c}}$ possible ways of sequencing these successes and failures.

    \Cref{eq:pcc} easily generalises to the binomial cumulative distribution function expressed in \Cref{eq:pc-less}---the probability of observing at most $\hat{c}$ successes is the sum of the probabilities of observing $j$ successes for $j=1,\dots,\hat{c}$.
    \begin{equation}\label{eq:pc-less}
        p_{c}(Y \leq \hat{c}) = \sum_{j=1}^{\hat{c}} \binom{c}{j} \bigr(p_1(X \geq 1)\bigr)^j \bigr(1-p_1(X \geq 1)\bigr)^{c-j}
    \end{equation}
    Therefore the probability of observing more than $\hat{c}$ successes is given by \Cref{eq:pc-raw} below, which expands as \Cref{eq:pc}.
    \begin{equation}\label{eq:pc-raw}
        p_{c}(Y > \hat{c}) = 1-p_{c}(Y \leq \hat{c})
    \end{equation}
\end{proof}

\begin{figure*}
    \centering
    \includegraphics[width=\linewidth,keepaspectratio]{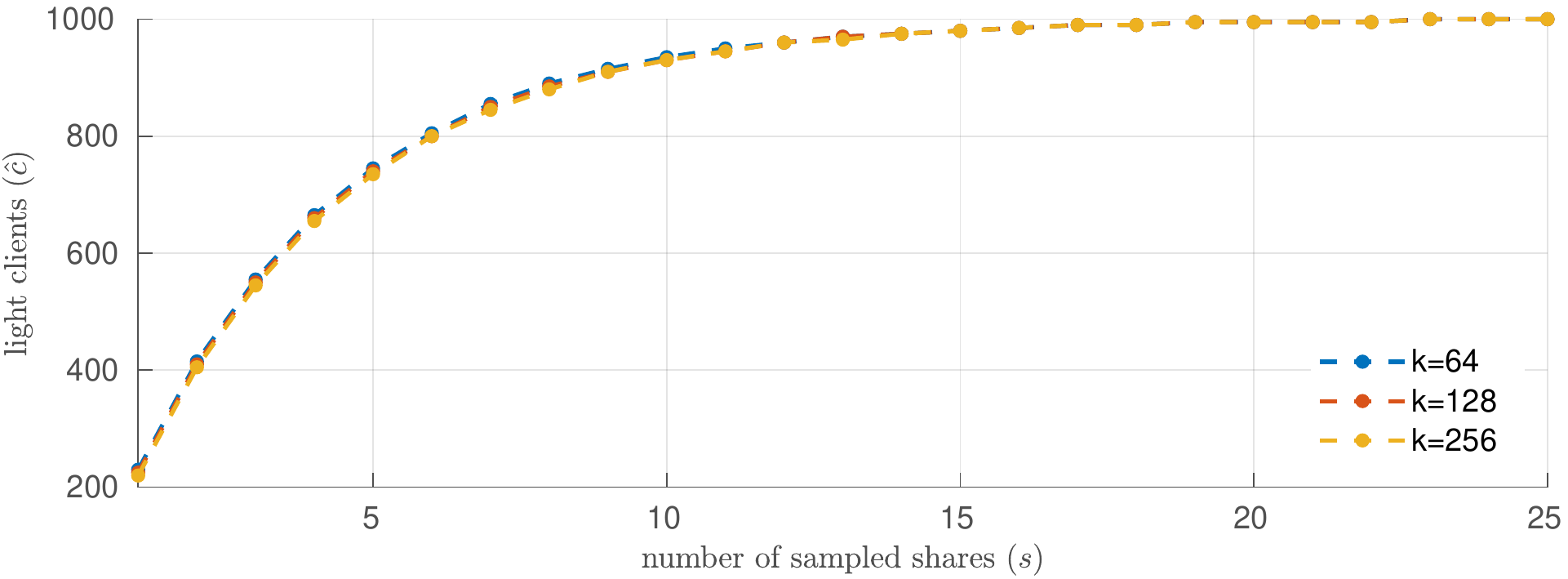}
    \caption{Plot of \Cref{eq:pc}---variation of the number of light clients $\hat{c}$ for which $p_{c}(Y > \hat{c}) \geq 0.99$ with the sampling size $s$. The total number of clients is fixed to $c=1000$, and the matrix sizes are $k=64,128,256$; \Cref{eq:pc} is however almost independent of $k$, as indicated by \Cref{eq:p1_lim}. }
    \label{fig:pc_s}
\end{figure*}

\Cref{fig:pc_s} shows the variation of the number of light clients $\hat{c}$ for which $p_{c}(Y > \hat{c}) \geq 0.99$ with the sampling size $s$. The total number of clients is fixed to $c=1000$, and the matrix sizes are $k=64,128,256$; \Cref{eq:pc} is however almost independent of $k$, as indicated by \Cref{eq:p1_lim}. This figure can be used to determine the number of light clients that will detect incomplete matrices with high probability ($p_{c}(Y > \hat{c}) \geq 0.99$), and that there is little gain in increasing $s$ over 15.

\subsubsection{Recovery and Selective Share Disclosure}\label{sec:recovery}

\Cref{co:matrix-lottery} presents the probability that light clients collectively samples enough shares to recover every share of the $2k \times 2k$ matrix.

If the light clients collectively sample all but $(k+1)^2$ distinct shares, the block producer cannot release any more shares without allowing the network to recover the whole matrix; it follows from \Cref{th:min-shares-unavailable} that light clients need to collect at least:
$$\gamma=(2k)^2 - (k+1)^2 + 1 = k(3k-2)$$
distinct shares (randomly chosen) to have the certainty to be able to recover the $2k \times 2k$ matrix. We are therefore interested in the probability that light clients---each sampling $s$ distinct shares---collectively samples at least $\gamma$ distinct shares; this probability is expressed by \Cref{co:matrix-lottery}.

\begin{theorem}\label{th:euler1785}
    (Euler~\cite{euler1785}) the probability that the number of distinct elements sampled from a set of $n$ elements, after $c$ drawings with replacement of $s$ distinct elements each, is at least all but $\lambda$ elements\footnote{This problem is also known as \emph{the coupon collector's problem} with group drawing~\cite{ferrante2014}.}:
    \begin{eqnarray}\label{eq:pe}
        p_{e}(Z \geq n-\lambda) &=& 1-\sum_{i=1}^\infty(-1)^i{\lambda+i-1\choose\lambda}{n\choose\lambda+i}\bigr(W_i\bigr)^c\\ \nonumber
        \textrm{where} && W_i={n-\lambda-i\choose s}\bigg/{n\choose s}
    \end{eqnarray}
\end{theorem}

\begin{corollary}\label{co:matrix-lottery}
    Given a $2k \times 2k$ matrix $E$ as shown in \Cref{fig:reedsolomon}, where each of $c$ players randomly samples $s$ distinct shares from $E$. The probability that the players collectively sample at least $\gamma=k(3k-2)$ distinct shares is $p_{e}(Z \geq \gamma)$
\end{corollary}
\begin{proof}
    \Cref{co:matrix-lottery} can be easily proven by substituting $\lambda=n-\gamma$ and $n=(2k)^2$ into \Cref{eq:pe}.
\end{proof}

\begin{figure*}
    \centering
    \begin{subfigure}[b]{0.48\textwidth}
        \includegraphics[width=\linewidth,keepaspectratio]{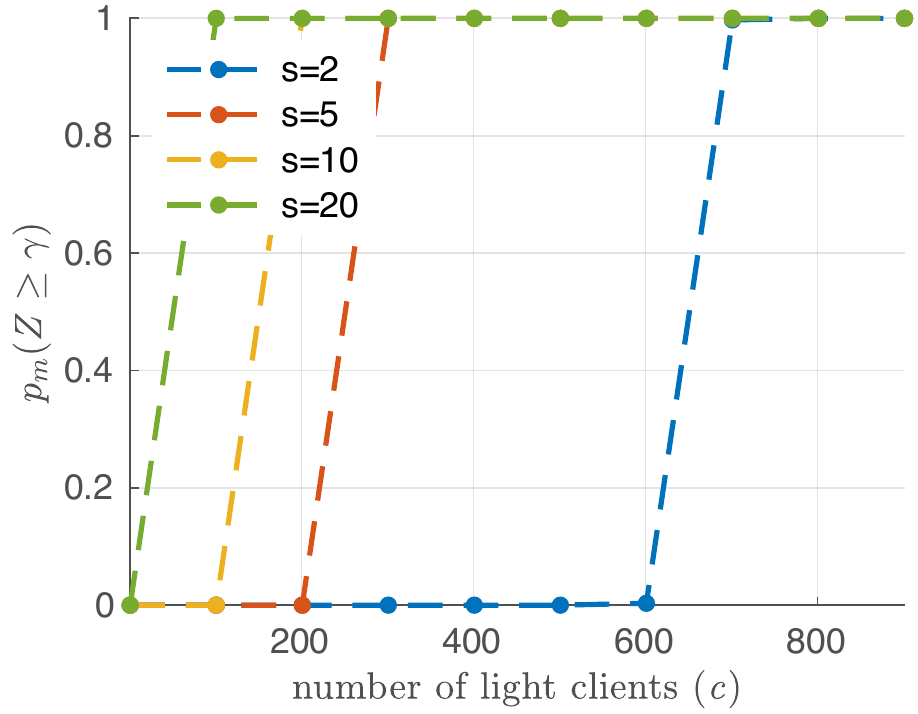}
        \caption{Matrix size $k=16$}
       \label{fig:pe_k=16}
    \end{subfigure}
    \quad
    \begin{subfigure}[b]{0.48\textwidth}
        \includegraphics[width=\linewidth,keepaspectratio]{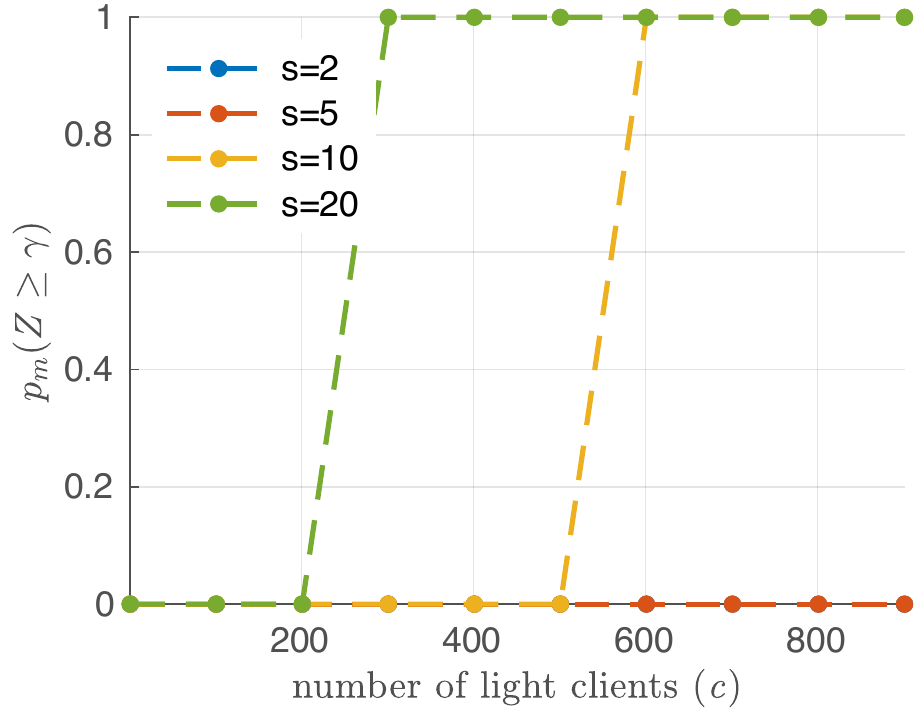}
        \caption{Matrix size $k=32$}
       \label{fig:pe_k=32}
    \end{subfigure}
    \caption{Plot of \Cref{co:matrix-lottery}---variation of the probability $p_{e}(Z \geq \gamma)$ with the number of clients $(c)$ for different values of $s$ and $k$.}
   \label{fig:pe}
\end{figure*}

\begin{table}[t]
    \begin{center}
        \begin{tabular}{L C C C C C}
            \toprule

            \small\bf\boldmath$p_{e}(Z \geq \gamma)$ &
            \small\bf\boldmath$s=2$ & \small\bf\boldmath$s=5$ &
            \small\bf\boldmath$s=10$ &
            \small\bf\boldmath$s=20$ &
            \small\bf\boldmath$s=50$\\

            \midrule

            \small\bf\boldmath$k=16$ & 692 & 277 & 138 & 69 & 28\\
            \small\bf\boldmath$k=32$ & 2805 & 1,122 & 561 & 280 & 112\\
            \small\bf\boldmath$k=64$ & 11,289 & 4,516 & 2,258 & 1,129 & 451\\
            \small\bf\boldmath$k=128$ &  $>$40,000 & $\sim$18,000 & $\sim$9,000 & $\sim$4,500 & 1,811\\

            \bottomrule
        \end{tabular}
    \end{center}
    \caption{Minimum number of light clients ($c$) required to achieve $p_{e}(Z \geq \gamma) > 0.99$ for various values of $k$ and $s$. The approximate values have been approached numerically as evaluating \Cref{eq:pe} can be extremely resource-intensive for large values of $k$.}
   \label{tab:timing}
\end{table}

Contrarily to \Cref{eq:pc}, \Cref{fig:pe} shows that $p_{e}(Z \geq \gamma)$ depends on the matrix size $k$.

\subsection{Properties Security Analysis}\label{sec:properties-security-analysis}

\subsubsection{Standard Model}

\begin{corollary}\label{th:cor_standard_model}
    Under the standard model, a block producer cannot cause soundness (\Cref{def:soundness}) and agreement (\Cref{def:agreement}) to fail for more than $c$ honest clients with a probability lower than $p_1(X \geq 1)$  per client, where $c$ is determined by the probability distribution $p_{e}(Z \geq \gamma)$.
\end{corollary}
\begin{proof}
    \Cref{co:matrix-lottery} shows that with probability $p_{e}(Z \geq \gamma)$, $c$ honest clients will sample enough shares to collectively recover the full block. Honest clients will gossip these shares to full nodes which then gossip them to each other, and within $k \times \delta$ at least one honest full node will then recover the full block data, thus satisfying soundness with a probability of $1 - p_1(X \geq 1)$ per client (the probability of the block producer not passing the client's random sampling challenge when all the block data is available).

    If the data is available and no fraud proofs of incorrectly generated extended data was received by the client, then no other client will receive a fraud proof either, due to our assumption that there is at least one honest full node in the network and honest light clients are not under an eclipse attack, thus satisfying agreement with a probability of $1 - p_1(X \geq 1)$ per client.

    Due to the selective share disclosure attack described in \Cref{sec:selective-share-disclosure}, this means that the block producer can violate soundness and agreement of the first $c$ clients that make sample requests, as the block producer can stop releasing shares just before it is about to release the final shares to allow the block to be recoverable.
\end{proof}

\subsubsection{Enhanced Model}

\begin{corollary}\label{th:cor_enhanced_model}
    Under the enhanced model, a block producer cannot cause soundness (\Cref{def:soundness}) and agreement (\Cref{def:agreement}) to fail with a probability lower than $p_x(X \geq 1)$  per client,
    \begin{equation}\label{eq:px}
        p_x(X \geq 1) = \sum_{i=1}^d \frac{\binom{s}{i}\binom{s(c-1)}{d-i}}{\binom{c \cdot s}{d}}
    \end{equation}
    where $c$ is the number of clients and $d$ is the number of requests that the block producer must deny to prevent full nodes from recovering the data.
\end{corollary}
\begin{proof}
    The proof of \Cref{th:cor_enhanced_model} starts as the proof of \Cref{th:cor_standard_model}; honest light clients collectively samples enough shares to recover the full block data by gossiping these shares to full nodes; soundness is satisfied with probability $1 - p_1(X \geq 1)$ per client. None of the light clients receive fraud proofs if the full data is available and no valid fraud proofs are sent over the network, and all light clients eventually receive a valid fraud proof if one is sent, satisfying agreement with the same probability.

    However, the enhanced model assumes that all sample requests come through a perfect mix network (\ie requests are unlinkable between each other), and defeats the selective shares disclosure attack presented in \Cref{sec:selective-share-disclosure}. The enhanced model removes the notion of `first' clients described in \Cref{th:cor_standard_model} as block producers cannot distinguish which requests comes from which client (since requests are unlikable). Furthermore, if block producers randomly deny some requests, light clients would uniformly see some of their sample requests denied, and each light client would therefore consider the block invalid with equal probability.

    Particularly, if $c$ light clients each sample $0 < s < (k+1)^2$ shares, block producers observe a total of $(c \cdot s)$ indistinguishable requests. Let us assume that a malicious block producer must deny at least $d$ request to prevent full nodes from recovering the block data.
    The probability that a light client observes at least one of its requests denied (and thus rejects the block) is given by $p_x(X \geq 1)$ in \Cref{eq:px}.
    The numerator of \Cref{eq:px} computes the number of ways of picking $i$ of the denied requests among the $s$ requests sent by the light client (i.e., $\binom{s}{i}$), multiplied by the number of ways to pick the remaining $d-i$ requests among the set of requests sent by other light clients: $c \cdot s - s = s(c-1)$ (i.e., $\binom{s(c-1)}{d-i}$). The denominator computes the total number of ways to pick any $d$ requests out of the total number of requests (i.e., $\binom{c \cdot s}{d}$). The probability that at least one of the denied requests comes from a particular client is the sum of the probabilities for $i=1,\dots,d$.
\end{proof}

Like \Cref{eq:p1}, \Cref{eq:px} rapidly grows and shows that light clients reject the block if invalid (for appropriate values of $d$). The value of $d$ can be approximated using \Cref{co:matrix-lottery}, and depends on $s$ and $c$. To provide a quick intuition, if we assume that the light clients collectively sample at least once every share of the block, a malicious block producer must deny at least $(k+1)^2$ requests on different shares to prevent full nodes from recovering the block data; since multiple requests can sample the same shares, $d \geq (k+1)^2$.

%% file: sections/implementation.tex
\section{Performance and Implementation}\label{sec:implementation}

We implemented the data availability proof scheme described in \Cref{sec:availability} and a prototype of the state transition fraud proof scheme described in \Cref{sec:fraud-proofs} in 2,683 lines of Go code and released the code as a series of free and open-source libraries.\footnote{2D Reed-Solomon Merkle tree data availability scheme: \url{https://github.com/musalbas/rsmt2d}\\
State transition fraud proofs prototype: \url{https://github.com/asonnino/fraudproofs-prototype}\\
Sparse Merkle tree library: \url{https://github.com/musalbas/smt}}

We first evaluate the space and time complexity of the scheme in \Cref{sec:complexity} and then present the performance benchmarks of our implementation in \Cref{sec:benchmarks}. We perform the measurements on a laptop with an Intel Core i5 1.3GHz processor and 16GB of RAM, and use SHA-256 for hashing.

\subsection{Space and Time Complexity}\label{sec:complexity}

\begin{table}[t]
    \begin{center}
        \begin{tabular}{l c}
            \toprule

            \small\bf Object &
            \small\bf Worst case space complexity \\

            \midrule

            \small\bf State fraud proof & $O(p + p\log(d) + w\log(s) + w)$\\
            \small\bf Availability fraud proof & $O(d^{0.5} + d^{0.5}\log(d^{0.5}))$\\
            \small\bf Single sample response & $O(\mathsf{shareSize} + \log(d))$\\
            \small\bf Block header with axis roots & $O(d^{0.5})$\\
            \small\bf Block header excl. axis roots & $O(1)$\\

            \bottomrule
        \end{tabular}
    \end{center}
    \caption{Worst case space complexity for various objects. $p$ represents the number of transactions in a period, $w$ represents the number of state witnesses for those transactions, $d$ is short for $\mathsf{dataLength}$, and $s$ is the number of key-value pairs in the state tree.}
    \label{tab:space-complexity}
\end{table}

\begin{table}[t]
    \begin{center}
        \begin{tabular}{l c}
            \toprule

            \small\bf Action &
            \small\bf Worst case time complexity \\

            \midrule

            \small\bf [G] State fraud proof & $O(b + p\log(d) + w)$\\
            \small\bf [V] State fraud proof & $O(p + p\log(d) + w)$\\
            \small\bf [G] Availability fraud proof & $O(d^2 + d^{0.5}\log(d^{0.5}))$\\
            \small\bf [V] Availability fraud proof & $O(d + d^{0.5}\log(d^{0.5}))$\\
            \small\bf [G] Availability fraud proof (FFTs) & $O(d\times d^{0.5}\log(d^{0.5}))$\\
            \small\bf [V] Availability fraud proof (FFTs) & $O(d^{0.5}\log(d^{0.5}))$\\
            \small\bf [G] Single sample response & $O(\log(d^{0.5}))$\\
            \small\bf [V] Single sample response & $O(\log(d^{0.5}))$\\

            \bottomrule
        \end{tabular}
    \end{center}
    \caption{Worst case time complexity for various actions, where [G] means generate and [V] means verify. $p$ represents the number of transactions in a period, $b$ represents the number of transactions in the block, $w$ represents the number of state witnesses for those transactions, $d$ is short for $\mathsf{dataLength}$, and $s$ is the number of key-value pairs in the state tree. For generating and verifying state fraud proofs, we assume that each transaction takes the same amount of time to process. For generating fraud proofs, we also include the cost of verifying the block itself.}
    \label{tab:time-complexity}
\end{table}

\Cref{tab:space-complexity} shows the space complexity for different objects. We observe that the size of the state transition fraud proofs only grows logarithmically with the size of the block and state, whereas the availability fraud proofs (as well as block headers with the axis roots) grows at least in proportion to the square root of the size of the block.

\Cref{tab:time-complexity} shows the time complexity for various actions. For generating and verifying fraud proofs, we note that generating and verifying Merkle proofs for state witnesses is $O(1)$ as a sparse Merkle tree has a static depth. The most expensive operation is generating availability fraud proofs, as Lagrange interpolation takes $O(k^2)$ time to encode or decode a row/column with $k$ shares, but this can be reduced to $O(k\log(k))$ time with algorithms based on Fast Fourier Transforms (FFT) \cite{lin2014,reed1978}.

\subsection{Benchmarks}\label{sec:benchmarks}

\begin{table}[t]
    \begin{center}
        \begin{tabular}{l c c}
            \toprule

            \small\bf Object (10 tx/period) &
            \small\bf Size ($\sim$0.25MB block) &
            \small\bf Size ($\sim$1MB block) \\

            \midrule

            \small\bf State fraud proof & 14,090b & 14,410b\\
            \small\bf Availability fraud proof & 12,320b & 26,688b\\
            \small\bf Single sample response & 320b & 368b\\
            \small\bf Block header with. axis roots & 2,176b & 4,224b\\
            \small\bf Block header excl. axis roots & 128b & 128b\\

            \bottomrule
        \end{tabular}
    \end{center}
    \caption{Illustrative sizes for objects for $\sim$0.25MB and $\sim$1MB blocks, assuming that a period consists of 10 transactions, the average transaction size is 225 bytes, and that conservatively there are $2^{30}$ non-default nodes in the state tree.}
    \label{tab:size-benchmarks}
\end{table}

\begin{table}[t]
    \begin{center}
        \begin{tabular}{l c c}
            \toprule

            \small\bf Action &
            \small\bf Time ($\sim$0.25MB block) &
            \small\bf Time ($\sim$1MB block) \\

            \midrule

            \small\bf [G] State fraud proof & 289.78 ms & 981.88 ms\\
            \small\bf [V] State fraud proof & 1.50 ms & 1.50 ms\\
            \small\bf [G] Availability fraud proof & 7.96ms & 50.88ms \\
            \small\bf [V] Availability fraud proof & 0.05ms & 0.19ms\\
            \small\bf [G] Single sample response & $<0.00001$ms & $<0.00001$ms\\
            \small\bf [V] Single sample response & $<0.00001$ms & $<0.00001$ms\\

            \bottomrule
        \end{tabular}
    \end{center}
    \caption{Computation time (mean over ten repeats) for various actions, where [G] means generate and [V] means verify. We assume that a period consists of 10 transactions, the average transaction size is 225 bytes, and each transaction writes to one key in the state tree.}
    \label{tab:time-benchmarks}
\end{table}

\Cref{tab:size-benchmarks} shows the size of various objects when transmitted over the network. We observe that the size of the state fraud proof only increases logarithmically with the size of the block; this is because the number of transactions in a period remains static, but the size of the Merkle proof for each transaction increases slightly. Block size impacts the size of availability fraud proofs and the axis roots the most, as the size of a single row or column is proportional to the square root of the size of the block.

\Cref{tab:time-benchmarks} shows the computation time for generating and verifying various objects; the benchmark for state fraud proof generation includes time spent verifying the block. Although verification is linear in the size of the block, in our implementation it has a high constant factor due to the need for 256 hash operations per update in the tree. This can be improved by using a SHA-256 library that uses SIMD instructions \cite{guilford2012} and splitting up the tree into sub-trees \cite{cutter2016} so that updates can be processed in parallel. Alternatively, a more complex key-value tree construction can be used such as a Patricia tree \cite{wood2018}.

As expected, verifying an availability fraud proof is significantly quicker than generating one. This is because generation requires checking the entire data matrix, whereas verification only requires checking one row or column. Note that we used a library that uses standard Reed-Solomon algorithms that take $O(k^2)$ time to encode/decode---the benchmarks can be improved by using FFT-based algorithms that take $O(k\log(k))$ time.

%% file: sections/discussion.tex
\section{Discussion}\label{sec:discussion}

\subsection{Succinct Proofs of Computation}\label{sec:succinct-proofs}

There have been advances in succinct proofs of computation, including zk-SNARKs \cite{ben-sasson2014} and more recently zk-STARKs \cite{ben-sasson2018}, which allow a prover to prove that $f(x, W) = y$ for some provided $x$ and $y$, where even if the witness $W$ is very large in size and the computation $f$ takes a very long time to compute, the proof itself has only logarithmic or constant size and takes logarithmic or constant time to verify.

For future work, we can require block headers to come with such a proof to show that they are correctly erasure coded, removing the need for fraud proofs. Also note that the only significant advantage of the 2D Reed Solomon scheme over the 1D scheme is smaller fraud proofs, so if succinct proofs are used switching back to 1D may be optimal (constructing a legitimate erasure code takes only $O(n\log(n))$ computation time for $n$ shares if Fast Fourier Transforms are used \cite{lin2014,reed1978}).

Note that while succinct proofs of computation of the block state root transition function can be used to remove the need for fraud proofs, they do not remove the need for availability proofs. If malicious majority broadcasts a block for which the data is not available, they can deny honest full nodes the information that they need to construct the full up-to-date state, and generate witnesses (\ie Merkle branches) for transactions touching certain accounts. By preventing witness creation, a block with unavailable data can make accounts permanently inaccessible. \cite{camacho2009} uses information-theoretic arguments to show that even constructions such as cryptographic accumulators cannot remove the need to verify availability of $O(n)$ data to ensure that all witnesses can be correctly updated.

\subsection{Locally decodable codes}

Another strategy for removing the need for fraud proofs from this scheme is to use the local decodability feature of multi-dimensional Reed-Solomon codes\cite{yekhanin2010}. Particularly, we construct a ``proof of proximity'' that consists of a set of pseudorandomly selected rows and columns (or axis-parallel lines more generally for higher-dimensional codes), using the Merkle root of the data as a source of entropy, which the verifier can verify have degree $< k$, thereby probabilistically verifying that a very high percentage of all axis-parallel lines have degree $< k$ and therefore any non-axis-parallel line has degree $< d * k$. The file would be extended from $k * ... * k$ to $(k*2d) * ... * (k*2d)$.

Any block (header) that comes with a valid proof of proximity is admissible, even if some small portion of the extended polynomial data is incorrect. Because of this, a single Merkle branch no longer suffices to prove a single value. Instead, the verifier can select a random non-axis-parallel line that passes through the point, and require the prover to provide at least $\frac{3}{2} * d$ points along the line. The verifier computes the correct value at the desired point, doing error correction if necessary. For added soundness, the verifier can select multiple random non-axis-parallel lines.

This scheme has the benefit that it does not rely on fraud proofs or expensive proofs of computation, but has the weaknesses that \first it requires more encoded data to be stored across the network, though this is mitigated by the fact that the larger number of shares makes it safer to have a smaller number of copies of each share stored across the network, and \second Merkle proofs become roughly two orders of magnitude larger.

%% file: sections/related-work.tex
\section{Related Work}\label{sec:related-work}

The original Bitcoin whitepaper \cite{nakamoto2008} briefly mentions the possibility of `alerts', which are messages sent by full nodes to alert light clients that a block is invalid, prompting them to download the full block to verify the inconsistency. Little further exploration has been done on this, partly due to the data availability problem.

There have been online discussions about how one may go about designing a fraud proof system \cite{ranvier2017, todd2016}, however no complete design that deals with all block invalidity cases and data availability has been proposed. These earlier systems have taken the approach of attempting to design a fraud proof for each possible way to create a block that violates the transaction validity rules (\eg double spending inputs, mining a block with a reward too high, etc), whereas this paper generalises the blockchain into a state transition system with only one fraud proof.

On the data availability side, Perard \etal \cite{perard2018} have proposed using erasure coding to allow light clients to voluntarily contribute to help storing the blockchain without having to download all of it, however they do not propose a scheme to allow light clients to verify that the data is available via random sampling and fraud proofs of incorrectly generated erasure codes.

%% file: sections/conclusion.tex
\section{Conclusion}\label{sec:conclusion}

We presented, implemented and evaluated a complete fraud and data availability proof scheme, which enables light clients to have security guarantees almost at the level of a full node, with the added assumptions that there is at least one honest full node in the network that distributes fraud proofs within a maximum network delay, and that there is a minimum number of light clients in the network to collectively recover blocks.

%% file: sections/acknowledgements.tex
\section*{Acknowledgements}

Mustafa Al-Bassam is supported by a scholarship from The Alan Turing Institute and Alberto Sonnino is supported by the European Commission Horizon 2020 DECODE project under grant agreement number 732546.

Thanks to George Danezis, Alexander Hicks and Sarah Meiklejohn for helpful discussions about the mathematical proofs.

%% file: sections/appendix.tex
\appendix

\section{Double-tree Design}

In the paper we have considered that block headers contain a single data root $\mathsf{dataRoot}_i$ that includes both transactions and intermediate state roots, in order to allow for data availability proofs. However, we can design a more simplified structure with two trees, $\mathsf{txRoot}_i$ and $\mathsf{txLength}_i$ for transactions and $\mathsf{traceRoot}_i$ and $\mathsf{traceLength}_i$ for intermediate state roots. Each leaf of the $\mathsf{txRoot}_i$ is a transaction and each leaf of the $\mathsf{traceRoot}_i$ is an intermediate state root. This does not require arranging data into fixed-size shares, however it does not support data availability proofs.

\subsection{Period Criterion}\label{sec:period-criterion}

The protocol may define a rule such that an intermediate state root must be added to the trace Merkle tree after a certain criterion has been met, for example after every $p$ transactions. We called this a `period criterion'. However unlike the period criterion mechanism described in \Cref{sec:shares-periods}, we assume that the period criterion under the double-tree design is a fixed number of transactions, so that it is possible for a fraud proof verified to know which trace is mapped to a specific transaction without downloading all the transactions (\ie there is a deterministic mapping between transaction indexes and trace indexes).

Based on the rule, we assume a function $\mathsf{period}(\mathsf{txIndex}) = \mathsf{traceIndex}$ that returns the index $\mathsf{traceIndex}$ of the intermediate trace root in the execution trace that is the pre-state for a transaction at index $\mathsf{txIndex}$ in the block's transaction list, or $-1$ if the pre-state if the previous block's $\mathsf{stateRoot}$. We denote $\mathsf{trace}_i^j$ as the $j$th intermediate state root for block $i$ in the tree committed to by $\mathsf{traceRoot}$. If $\mathsf{period}(\mathsf{txIndex}) = \mathsf{traceIndex}$, the pre-state root is $\mathsf{trace}_i^{\mathsf{traceIndex}}$ if $\mathsf{traceIndex} \geq 0$, or $\mathsf{stateRoot}$ if $\mathsf{traceIndex} = -1$.

A standard implementation of $\mathsf{period}$ may be $\mathsf{period}(\mathsf{txIndex}) = \lfloor\frac{\mathsf{txIndex}}{p}\rfloor - 1$ if there is a trace every $p$ transactions.

\subsection{Proof of Invalid State Transition}

A miner may incorrectly compute $\mathsf{stateRoot}_i$, for example by placing a series of random bytes as $\mathsf{stateRoot}_i$, or by crafting a malicious $\mathsf{stateRoot}_i$ that modifies the state in an invalid way. We can thus use the execution trace provided by $\mathsf{traceRoot}_i$ to prove that some part of the execution trace was invalid.

We define a function $\mathsf{VerifyTransitionFraudProof}$ and its parameters which verifies fraud proofs received from full nodes. If the fraud proof is valid, then the block that the fraud proof is for is permanently rejected by the client.

\begin{align*}
&\mathsf{VerifyTransitionFraudProof}(\mathsf{blockHash}_i,\\
    &\indent \mathsf{trace}_i^x, \{\mathsf{trace}_i^x \rightarrow \mathsf{traceRoot}_i\}, x, \tag*{(pre-state root)}\\
    &\indent \mathsf{trace}_i^{x+1}, \{\mathsf{trace}_i^{x+1} \rightarrow \mathsf{traceRoot}_i\}, \tag*{(post-state root)}\\
    &\indent(t_i^y, t_i^{y+1}, ..., t_i^{y+m}), y, \tag*{(transactions)}\\
    &\indent(\{t_i^y \rightarrow \mathsf{txRoot}_i\}, \{t_i^{y+1} \rightarrow \mathsf{txRoot}_i\}\}, ..., \{t_i^{y+m} \rightarrow \mathsf{txRoot}_i\}\}),\\
    &\indent(w_i^y, w_i^{y+1}, ..., w_i^{y+m}), \tag*{(state witnesses)}\\
&) \in \{\mathsf{true}, \mathsf{false}\}
\end{align*}

The pre-state root may be omitted from the fraud proof parameters if it is simply the state root of the previous block, and the post-state root may be omitted if it is the state root of the current block, as the client already knows these roots as they are in the block headers.

$\mathsf{VerifyTransitionFraudProof}$ returns $\mathsf{true}$ if all of the following conditions are met, otherwise $\mathsf{false}$ is returned:
\begin{enumerate}
    \item $\mathsf{blockHash}_i$ corresponds to a block header $h_i$ that the client has downloaded and stored.
    \item $\mathsf{VerifyMerkleProof}(\mathsf{trace}_i^x, \{\mathsf{trace}_i^x \rightarrow \mathsf{traceRoot}_i\}, \mathsf{traceRoot}_i, \mathsf{traceLength}_i, x)$ returns $\mathsf{true}$ if a pre-state root is specified.
    \item $\mathsf{VerifyMerkleProof}(\mathsf{trace}_i^{x+1}, \{\mathsf{trace}_i^{x+1} \rightarrow \mathsf{traceRoot}_i\}, \mathsf{traceRoot}_i, \mathsf{traceLength}_i,\allowbreak x + 1)$ returns $\mathsf{true}$ if a post-state root is specified.
    \item For each transaction $t_i^{y+a}$ in the proof, $\mathsf{period}(y+a) = x$ is true if a pre-state root is specified, otherwise $\mathsf{period}(y+a) = -1$ and $y = 0$ is true.
    \item For each transaction $t_i^{y+a}$ in the proof, $\mathsf{VerifyMerkleProof}(t_i^{y+a}, \{t_i^{y+a} \rightarrow \mathsf{txRoot}_i\},\allowbreak \mathsf{txRoot}_i, \mathsf{txLength}_i, y+a)$ returns $\mathsf{true}$.
    \item Let the intermediate state roots after applying every transaction in the proof one at a time be $\mathsf{interRoot}_i^j = \mathsf{rootTransition}(\mathsf{interRoot}_i^{j-1}, t_i^j, w_i^j)$. If a pre-state root is specfied, then the base case is $\mathsf{interRoot}_i^{y} = \mathsf{trace}_i^x$, otherwise $\mathsf{interRoot}_i^{y} = \mathsf{stateRoot}_{i-1}$. If a post-state is specified, $\mathsf{trace}_i^{x+1} = \mathsf{interRoot}_i^{y+m}$ is true, otherwise $\mathsf{stateRoot}_{i} = \mathsf{interRoot}_i^{y+m}$ and $y+m = \mathsf{txLength}_i$ is true.
\end{enumerate}

\section{Computation of $\mathsf{index}$ in Step 4 of $\mathsf{VerifyCodecFraudProof}$}\label{sec:computation-index}

In Step 4 of $\mathsf{VerifyCodecFraudProof}$ in \Cref{sec:fraud-proofs-extended-data}, $\mathsf{index}$ can be computed as follows:

\begin{itemize}
    \item If $\mathsf{axis} = 0$ and $\mathsf{ax}_x = 0$, $\mathsf{index} = j * \mathsf{matrixWidth}_i + \mathsf{pos}_x$.
    \item If $\mathsf{axis} = 1$ and $\mathsf{ax}_x = 0$, $\mathsf{index} = \mathsf{pos}_x * \mathsf{matrixWidth}_i + j$.
    \item If $\mathsf{axis} = 1$ and $\mathsf{ax}_x = 1$, $\mathsf{index} = \frac{1}{2}{\mathsf{dataLength_i}} + j * \mathsf{matrixWidth}_i + \mathsf{pos}_x$.
    \item If $\mathsf{axis} = 0$ and $\mathsf{ax}_x = 1$, $\mathsf{index} = \frac{1}{2}{\mathsf{dataLength_i}} + \mathsf{pos}_x * \mathsf{matrixWidth}_i + j$.
\end{itemize}